\documentclass[preprint,11pt,times]{elsarticle}

\usepackage[top=2.5cm,bottom=2.5cm,left=2.0cm,right=2.0cm]{geometry}
\usepackage{amssymb}
\usepackage{amsmath}
\newtheorem{proposition}{Proposition}
\newproof{proof}{Proof}
\usepackage{hyperref}
\usepackage{booktabs}
\usepackage{multirow}
\usepackage[dvipsnames]{xcolor}
\usepackage{adjustbox,caption}
\usepackage{array}
\usepackage{ragged2e}
\usepackage{placeins} 
\usepackage{hyphenat} 

\newcommand{\koutofn}{\(k\)-out-of-\(n\)}
\DeclareMathOperator*{\argmax}{arg\,max}
\DeclareMathOperator*{\argmin}{arg\,min}

\biboptions{sort&compress}

\begin{document}

\begin{frontmatter}

\title{The price of decentralization in managing engineering systems \\ through multi-agent reinforcement learning}

\author[label1]{Prateek Bhustali\texorpdfstring{\corref{cor1}}{*}}
\ead{p.bhustali@tudelft.nl}
\author[label1,label2]{Pablo G. Morato}
\author[label3]{Konstantinos G. Papakonstantinou}
\author[label1]{Charalampos P. Andriotis}
\ead{c.andriotis@tudelft.nl}

\affiliation[label1]{organization={Faculty of Architecture and the Built Environment, Delft University of Technology},
            city={Delft},
            postcode={2628 BL}, 
            country={The Netherlands}}
\affiliation[label2]{organization={Engineering Risk Analysis Group, Technical University of Munich},
            city={Munich},
            postcode={80333}, 
            country={Germany}}
\affiliation[label3]{organization={Department of Civil \& Environmental Engineering, The Pennsylvania State University},
            addressline={University Park}, 
            state={PA},
            postcode={16802}, 
            country={USA}}

\cortext[cor1]{Corresponding author.}

\begin{abstract}
Inspection and maintenance (I\&M) planning involves sequential decision making under uncertainties and incomplete information, and can be modeled as a partially observable Markov decision process (POMDP). While single-agent deep reinforcement learning provides approximate solutions to POMDPs, it does not scale well in multi-component systems. Multi-agent deep reinforcement learning (MADRL) improves scalability by decentralizing decision-making across agents that locally control individual components.
However, this decentralization can induce cooperation pathologies that degrade the optimality of the learned policies. To examine these effects in I\&M planning, we introduce a set of deteriorating systems in which redundancy is varied systematically. These benchmark environments are designed such that computation of centralized (near-)optimal policies remains tractable, enabling us to quantify the \emph{price of decentralization}. We implement and benchmark a broad set of MADRL algorithms spanning fully centralized and decentralized training paradigms, from value-factorization to actor--critic methods.
Our results show a clear effect of redundancy on coordination: MADRL algorithms achieve near-optimal performance in series-like settings, whereas increasing redundancy amplifies coordination challenges and leads to optimality losses averaging 10.3\% in the parallel configuration. Nonetheless, decentralized agents learn policies that consistently outperform optimized heuristic baselines. While both remain suboptimal relative to the centralized optimum, decentralized policies narrow the gap, highlighting the promise and current limitations of decentralized learning for maintenance planning.
\end{abstract}


\end{frontmatter}

\section{Introduction}

\begin{figure}[t] 
    \centering
    \includegraphics[width=\linewidth]{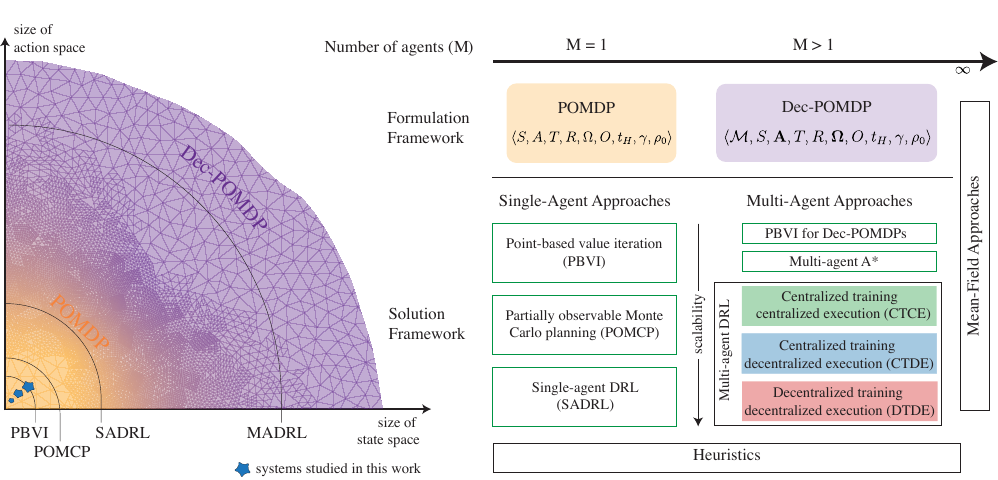}
    \caption{In multi-component systems, joint state, action, and observation spaces grow exponentially with the number of components, rendering single-agent approaches intractable. Dec-POMDPs address this curse of dimensionality by decentralizing control and are commonly solved using multi-agent DRL methods. 
    \textit{Left}: Scalability of solution approaches. The curved boundaries illustrate the approximate, relative scalability of POMDP and Dec-POMDP solution approaches as these spaces increase. The (\includegraphics[height=1.2ex]{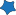}) indicates the approximate size of the systems studied in this work, where single-agent methods remain tractable and allow direct comparison with multi-agent approaches. 
    \textit{Right}: Agent-centric paradigms. Single-agent approaches are intrinsically centralized because a single agent possesses all information during training and execution. In contrast, multi-agent approaches, such as CTDE, relax decentralization during training by allowing agents to share information, but agents must learn to execute actions independently at execution/inference.}
    \label{fig:landscape}
\end{figure}

Large-scale infrastructure systems, such as electricity grids, road networks, and communication systems, require effective long-term inspection and maintenance (I\&M) planning for sustained reliability. Partially observable Markov decision processes (POMDPs) provide a systematic formulation for I\&M planning for such engineering systems under performance, data, and degradation uncertainties. POMDPs address several challenges inherent to the problem, such as stochastically changing component properties, partial observability of damage states, and long planning horizons~\citep{madanat_optimal_1994,corotis_modeling_2005,scharpff_planning_2013,papakonstantinou_planning_2014,Papakonstantinou2018POMDPAM,arcieri_bridging_2023,hu_condition-based_2026}. 
In multi-component systems, exact solutions to POMDPs are generally computationally intractable, motivating a variety of approximate approaches to I\&M planning. These include (i) heuristics that parameterize the policy space using risk-based or condition-based rules to keep the search space tractable~\citep{bismut_optimal_2021,luque_risk-based_2019}, (ii) online planning methods that perform approximate belief-space look-ahead using an environment simulator, such as partially observable Monte Carlo planning (POMCP) \citep{silver_monte-carlo_2010}, and (iii) offline planning methods that approximate optimal behavior over the belief space, including point-based value iteration methods (PBVI) \citep{spaan_perseus_2005,kurniawati_sarsop_2008} or deep reinforcement learning (DRL) \citep{mnih_playing_2013,schulman_proximal_2017,haarnoja_soft_2018,gallici_simplifying_2024}.

However, for multi-component systems, the joint state-, action-, and observation-spaces grow \textit{exponentially} with the number of components, rendering near-optimum POMDP solutions intractable for large-scale real-world I\&M planning problems. To alleviate this curse of dimensionality, decentralized POMDPs (Dec-POMDPs) extend the POMDP framework by designating an agent for each component or subsystem. This \textit{decentralization} effectively renders the problem as a cooperative multi-agent task. Similarly, this central idea is applied to solution methods, giving rise to multi-agent DRL (MADRL)~\citep{albrecht_multi-agent_2023,amato_first_2024}. Mean-field approaches, such as mean-field reinforcement learning~\citep{mean_field_RL}, extend this idea further to large agent populations by modeling interactions between individual agents and macroscopic population statistics.
Fig.~\ref{fig:landscape} highlights the relative scalability of different solution approaches, from single-agent, point-based POMDP solvers to MADRL. While exact and heuristic search methods, such as PBVI-based approaches for Dec-POMDPs~\citep{macdermed_point_2013} and multi-agent A*~\citep{koops_approximate_2024}, extend planning and search-based techniques to the multi-agent setting, they remain limited to small problem instances due to the exponential growth of the joint action and belief spaces. In contrast, MADRL relaxes optimality guarantees in exchange for improved scalability. Following the success of MADRL, several works within I\&M planning have adopted this paradigm to address scalability \cite{andriotis_managing_2019,andriotis_deep_2021,leroy_imp-marl_2023,saifullah_multi-agent_2024,DO2024123144,LEE2023109512,Bhattacharya2025,Arcieri2025}.

While decentralization relaxations in reinforcement learning alleviate scalability challenges, they can also introduce pathologies related to multi-agent learning, such as shadowed equilibria, environment non-stationarity, Pareto-selection issues, and multi-agent credit assignment, which can impede the convergence to optimal policies~\citep{matignon_independent_2012,claus_dynamics_1998}. 
These learning pathologies give rise to what we term the \emph{price of decentralization}: the relative loss in optimality incurred when decentralized control is used to achieve scalability in multi-agent decision-making problems.
As the underlying reward model and environment dynamics influence the strengths of these pathologies~\citep{claus_dynamics_1998}, their effect on the learned policies in I\&M planning problems remains to be understood. To systematically understand these implications, we contrast the three main MADRL paradigms: centralized training centralized execution (CTCE), centralized training decentralized execution (CTDE), and decentralized training decentralized execution (DTDE). 

Existing works benchmark the performance of MADRL algorithms~\citep{berlingerio_cooperative_2019,papoudakis_benchmarking_2021,leroy_imp-marl_2023}, including systems with up to 100 components~\citep{leroy_imp-marl_2023}. However, these evaluations rely on heuristic baselines, which provide no optimality guarantees and vary considerably in performance across problem settings. While these studies demonstrate the scalability of decentralized approaches, they cannot directly quantify the optimality losses caused by decentralization. We therefore study systems for which a near-optimal centralized solution can be computed, allowing us to measure the price of decentralization against a principled reference. Such losses may otherwise be obscured when heuristics are used as the sole baseline. This work builds on \citep{oliehoek_assessing_2025} by introducing new environments and a broader set of MADRL agents to examine multi-agent learning pathologies in I\&M planning. A detailed comparison with existing benchmarks and solution approaches is provided in Section~\ref{section:related work}.

To address these limitations in tracing the price of decentralization, we introduce benchmark environments based on the \koutofn system, focusing on a system with \(n=4\) components while varying the redundancy parameter \(k\) from 1 (parallel) to \(n\) (series)~\citep{kapur_reliability_2014}. These environments capture key characteristics of real-world deteriorating systems while remaining tractable for centralized algorithms. They include deterioration dynamics (modeled using upper-triangular matrices), redundancy, inspection and repair actions, partial observability of damage states, and heterogeneous components. In these environments, agents must jointly plan actions to minimize long-term inspection and maintenance costs over an infinite horizon. We adopt a discounted infinite-horizon formulation. This avoids the explicit time augmentation required to use finite-horizon variants with discounted infinite-horizon solvers. These formulations typically encode the time step in the state, which multiplies the state space by the planning horizon and increases planning complexity. To verify consistency of our findings, we also study smaller systems, with \(n=2\) and \(n=3\) components, defined analogously as subsets of the \(n=4\) system. As shown in Fig.~\ref{fig:landscape}, the small system sizes place the problem within the tractable regime of single-agent methods, allowing us to compute (near-)optimal solutions with SARSOP \cite{kurniawati_sarsop_2008}, a point-based POMDP solver with performance guarantees, and rigorously contrast the performance of DRL algorithms and heuristics. Using this framework, we assess the optimality of commonly used MADRL algorithms, including value-based and actor–critic, as well as on-policy and off-policy approaches.

We find that performance remains close to optimal in series and series-like systems; however, as redundancy increases (i.e., lower \(k\)) the performance of decentralized agents begins to deviate from the optimal, resulting in optimality losses. In particular, the mean optimality gap across decentralized algorithms increases from $0.3\% \pm 0.4\%$ in the 4-out-of-4 series system to $10.3\% \pm 3.9\%$ in the highly redundant 1-out-of-4 system. To further verify this trend, we conduct ablations that remove mobilization costs, i.e., fixed costs incurred when initiating an intervention,  and vary the failure penalties, i.e., costs incurred upon system failure, where we observe the same effect. As such, we provide insights into key limitations of current MADRL approaches in the I\&M context and provide a benchmark for diagnosing these challenges, informing the development of more effective algorithms. To complement the empirical findings, we analyze a simplified I\&M matrix game that isolates the coordination structure induced by redundancy. This analysis shows that methods based on value decomposition, which approximate joint action values as (monotonically) additive combinations of per-agent values, can be suboptimal in the presence of redundancy. Although derived in a single-step setting, it provides intuition for why such methods may struggle in systems with redundancies. Finally, despite the known pathologies of decentralized learning, some decentralized algorithms perform substantially better than others in challenging settings. In the 1-out-of-4 system, the best-performing IACC-PS policy reduces the optimality gap to 6.3\% compared with an average gap of 10.3\% across the tested algorithms.

\newpage
The main contributions of this work can be summarized as follows:
\begin{itemize}
\item We systematically compare the three main multi-agent learning paradigms and show that while decentralization achieves near-optimal performance in series and series-like systems, increasing redundancy induces coordination challenges that can lead to optimality losses.
\item We introduce benchmark environments for I\&M planning with (near-)optimal baselines. The benchmarks are based on \koutofn{} systems with \(n=4\) components and varying \(k \in \{1, 2,3,4\}\) and capture key features of real-world systems, including deterioration, redundancy, inspection and repair actions, partial observability, and component heterogeneity.
\item We provide open-source benchmark environments, (near-)optimal baselines, and DRL implementations, enabling reproducible comparisons across the proposed benchmarks.
\end{itemize}

Our benchmark environments, algorithm implementations, baselines, and trained models for reproducing the results of this paper are publicly available at
\href{https://github.com/prateekbhustali/price-of-decentralization}{\texttt{github.com/prateekbhustali/price-of-decentralization}}.

\subsection{Related work}
\label{section:related work}
To devise inspection and maintenance planning policies, conventional approaches use policy search within predefined policy subspaces. Exhaustive search over the feasible policy space is infeasible even for small problems. These subspaces are defined by heuristics such as risk-based thresholds, periodic inspection, or condition-based maintenance~\cite{grall_condition-based_2002,luque_risk-based_2019,bismut_optimal_2021}. However, these heuristics are often shown to be outperformed by DRL methods and may exclude the optimal policy from the searched policy subspaces.

Decentralized approaches have shown promising results in benchmark problems~\citep{papoudakis_benchmarking_2021,yu_surprising_2022,rashid_qmix_2018}. On the one hand, many recent works in I\&M planning address the scalability challenge by articulating the problem as a Dec-POMDP and solving it with MADRL algorithms~\citep{andriotis_managing_2019,leroy_imp-marl_2023,saifullah_multi-agent_2024,molaioni2024dynamic,DO2024123144,LEE2023109512}. 
For example, \citep{andriotis_managing_2019} alleviates the exponential growth of the policy output by assuming conditional independence of actions given the current action--observation history of agents (or joint beliefs). On the other hand, works such as \citep{hu_condition-based_2026} introduce a component-wise POMDP formulation. It decomposes the system-level decision problem into per-component POMDPs while retaining system-level dependencies (e.g., mobilization costs). This decomposition enables the use of POMDP solution methods shown in Fig.~\ref{fig:landscape}.

Prior comparisons of MADRL algorithms across the CTCE, CTDE, and DTDE paradigms for solving Dec-POMDPs can be found in~\citep{berlingerio_cooperative_2019}. They benchmarked MADRL algorithms on problems with relatively small state/\allowbreak action spaces and investigated policy gradient methods. In \citep{lyu_centralized_2023}, centralized and decentralized critics are compared, showing theoretically and empirically that centralized critics are not always beneficial. More recently, \citep{papoudakis_benchmarking_2021} extensively benchmarked MADRL algorithms, both value-based and policy gradient-based, on a diverse set of environments. In contrast, this work benchmarks MADRL algorithms across CTCE, CTDE, and DTDE paradigms while comparing against (near-)optimal baselines. Moreover, we examine the effect of redundancy on multi-agent cooperation in the context of I\&M planning.

To examine the efficacy of MADRL algorithms specifically for I\&M planning, \citep{leroy_imp-marl_2023} benchmarks the performance of state-of-the-art CTDE and DTDE algorithms in large-scale systems with up to 100 components. While they address I\&M planning at this scale, their evaluation relies on heuristic baselines that offer no guarantees on optimality. Moreover, their study does not specifically examine the role of redundancy in shaping multi-agent learning performance.

Recent work by~\cite{oliehoek_assessing_2025} provides a comparison of centralized and decentralized MADRL formulations for I\&M planning in finite-horizon \(k\)-out-of-\(5\) systems. The study evaluated off-policy methods against heuristic baselines and also reported that redundancy can hinder coordination. We now build on this foundation with the objective to explicitly quantify the optimality loss induced by decentralization. To this end, we introduce infinite-horizon \koutofn{} benchmarks with \(n\in\{2,3,4\}\), for which centralized (near-)optimal SARSOP policies remain tractable, and evaluate a broader set of on-policy and off-policy algorithms spanning CTCE, CTDE, and DTDE. This enables direct measurement of the price of decentralization and reveals how redundancy-induced degradation varies across algorithmic families and cost specifications. Finally, our analysis theoretically connects value-factorization limitations to redundancy in reliability systems.
\section{Background: Formulation and solution frameworks}

In this section, we briefly outline the mathematical frameworks leveraged in this work and concretely connect these frameworks to the I\&M problem in Section \ref{sec:I&M environments}.

\subsection{Partially observable Markov decision process (POMDP)}

POMDPs are a framework for modeling sequential decision-making in stochastic environments under partial observability. They are defined by the tuple $\langle S, A, T, R, \Omega, O, t_H, \gamma, \rho_0 \rangle$, where $S$ is the state space, $A$ is the action space, $T: S \times A \times S \to [0,1]$ is the transition model, $R: S \times A \to \mathbb{R}$ is the reward model, $\Omega$ is the observation space, $O: S \times A \times \Omega \to [0,1]$ is the observation model, $t_H \in \mathbb{N}_0 \cup\{\infty\}$ is the time horizon, $\gamma \in [0, 1]$ is the discount factor, and \(\rho_0\) is the initial distribution over states~\citep{kochenderfer_algorithms_2022,kaelbling_partially_1995}. The solution to a POMDP is an optimal policy $\pi^*$ that maximizes the expected sum of (discounted) rewards\footnote{In finite-horizon settings, $\gamma$ can be set to 1. However, sometimes a discount factor is also specified in the finite-horizon case to model the effect of inflation \citep{dewanto_examining_2022}. To avoid confusion, such effects can be subsumed into the reward model.}:
\begin{equation}
     \pi^* \in 
     \argmax_{\pi} J(\pi),
     \quad
    J(\pi) \coloneqq 
    \mathbb{E} \left[
    \sum_{t=0}^{t_H-1}\gamma^t R(s^t,a^t)
    \,\middle|\,
    \begin{array}{l}
    s^0 \sim \rho_0, \\
    a^t \sim \pi(\cdot\mid h^t),\\
    s^{t+1}\sim T(\cdot \mid s^t,a^t),\\
    o^{t+1}\sim O(\cdot\mid s^{t+1}, a^{t})
    \end{array}
    \right].
\label{eq:POMDP_optimization}  
\end{equation}
The policy maps the agent's action--observation history $h$ to the action(s)  $a \in A$. It can be deterministic, $a = \pi(h)$, or stochastic, $a \sim \pi(h)$. Although this initialization is often left implicit in POMDP formalisms, to make the history definition well-defined at \(t=0\), we introduce a null action
\(a^{-1}\coloneqq \phi \notin A\) and an initial observation \(o^0 \sim O^0(\cdot\mid s^0)\), so that \(h^0 \coloneqq \langle a^{-1},o^0 \rangle\). When we have access to transition and observation models, we can apply Bayes' rule to succinctly summarize $h$ into a probability measure called belief, $b$, representing uncertainty over the state space. This can effectively transform the POMDP into a \textit{belief-MDP}~\citep{astrom_optimal_1965}. We denote the initial belief by \(b_0 \coloneqq \rho_0\). For POMDPs with discrete states, the belief is a vector representing the probability over states, denoted by \(\vec b\). In this case, the belief about being in state $s^{\prime} \in S$ when we receive an observation $o$ after taking an action $a$ is:
\begin{equation}
    b^{a, o}(s') = \frac{O\left(o \mid s^{\prime}, a\right) \sum_{s \in S} T\left(s^{\prime} \mid s, a\right) b(s)}{\operatorname{Pr}(o \mid \vec b, a)},
    \label{eq:belief_update}
\end{equation}
where \(\operatorname{Pr}(o\mid\vec b, a)\) is a normalizing constant. Since beliefs are a \textit{sufficient} statistic summarizing the action--observation history, policies can map the belief space $B$ to actions, \(
\pi: B \to \operatorname{Pr}(A)\).

\subsection{Point-based POMDP solvers} \label{sec:POMDP solvers}

The principles of dynamic programming enable offline planning (also known as background planning) for MDPs when closed-form transition and reward models are available, $T: S \times A \times S \to [0, 1]$ and $R: S \times A \to \mathbb{R}$,  respectively~\citep{puterman2014markov}. When extended to POMDPs with known observation models, the problem can be reformulated as a belief MDP, where each belief is a probability distribution over system states. Although the belief space is continuous, its value function is convex and can be represented as the upper envelope of a finite set of linear functions, known as \(\alpha\)-vectors. Each \(\alpha\)-vector corresponds to a conditional plan and defines a linear function over the belief space. In principle, this representation allows \textit{exact value iteration}; however, it is computationally infeasible for large problems, due to the exponential growth in the number of \(\alpha\)-vectors ~\citep{shani_survey_2013}.

Point-based value iteration (PBVI), a class of approximate POMDP solvers, mitigates this curse of dimensionality using two key ideas. First, it represents the value function over a finite subset of belief space points. Second, it optimizes the value function approximation using point-based backups at those beliefs~\citep{shani_survey_2013}. Prominent PBVI solvers include Perseus~\citep{spaan_perseus_2005}, HSVI~\citep{hsvi}, and SARSOP~\citep{kurniawati_sarsop_2008}. Since, in general, a queried belief may not coincide with one of the beliefs backed up during offline planning, a look-ahead search can refine action selection by evaluating the one-step consequences at that belief.

In this work, we use SARSOP (short for successive approximations of the reachable space under optimal policies) to obtain near-optimal baselines. Unlike other solvers that approximate only a lower bound \( V^L(\vec b)\), SARSOP also maintains an upper bound on the value function~\citep{kurniawati_sarsop_2008}; when these bounds converge, they yield the optimal value \(V^L(\vec b) \leq V^{*}(\vec b)\). SARSOP returns a finite set of discrete \(\alpha\)-vectors \(\Gamma \coloneqq \{\vec \alpha_1, \vec \alpha_2, \ldots, \vec \alpha_{\zeta}\} \), each labeled with an action. The lower-bound approximation and the corresponding greedy policy, which selects the action associated with the maximizing vector, are given by:
\begin{equation}
V^{L}(\vec b) = \max_{\alpha \in \Gamma} \, \sum_{s \in S} \alpha(s) \, b (s),
\quad 
\text{and}
\quad
\pi^{\Gamma}_{\text{greedy}}(\vec b) 
= \text{action}\!\left(\argmax_{\alpha \in \Gamma} \ \sum_{s \in S} \alpha(s)\, b(s)\right).
\end{equation}
While \(V^L(\vec b)\) provides strong baselines for \(V^{*}(\vec b)\), it is computed from backups on a finite set of beliefs that aims to approximate the optimally reachable belief space \(\mathcal{R}^\ast(b_0)\). Therefore, \(\pi^{\Gamma}_{\text{greedy}}(\vec b)\) might not yield optimal actions for \(\vec b \notin \mathcal{R}^\ast(b_0)\) or when SARSOP has not converged. To mitigate these limitations, we improve the SARSOP policy with a one-step look-ahead search, which serves as a (near-)optimal baseline as follows (refer to \citep{kochenderfer_algorithms_2022} for further details):
\begin{equation}
\pi^{\Gamma}_{\text{look-ahead}}(\vec b) = \argmax_{a} \Bigg[ \sum_{s \in S} R(s,a) \, b(s) + \gamma \sum_{o\in\Omega} \Pr(o \mid \vec b,a) \, V^{L}(\vec b^{a,o}) \Bigg],
\end{equation}
where \(\Pr(o \mid \vec b,a) = \sum_{s \in S} b(s) \bigg [\sum_{s^{\prime} \in S} T(s^{\prime} \mid s,a) \, O(o \mid s^{\prime},a) \bigg ]\) is the probability of observing \(o \in \Omega\) after taking action \(a \in A\) 
when the current belief is \(\vec b\). Since the true state is unknown, we average over all possible current states \(s\) weighted by \(b(s)\). We then marginalize over all possible future transitions \(s^\prime\) in which \(o\) might be observed. For each \(s\), the
system transitions to the next state \(s^\prime\) with probability \(T(s^{\prime} \mid s,a)\). \(\vec b^{a,o}\) denotes the updated belief obtained via Eq.~\eqref{eq:belief_update}.

\subsection{Decentralized partially observable Markov decision process (Dec-POMDP)}
\label{sec:background_dec_pomdp}

\begin{figure}[t]
    \centering
    \includegraphics[width=\linewidth]{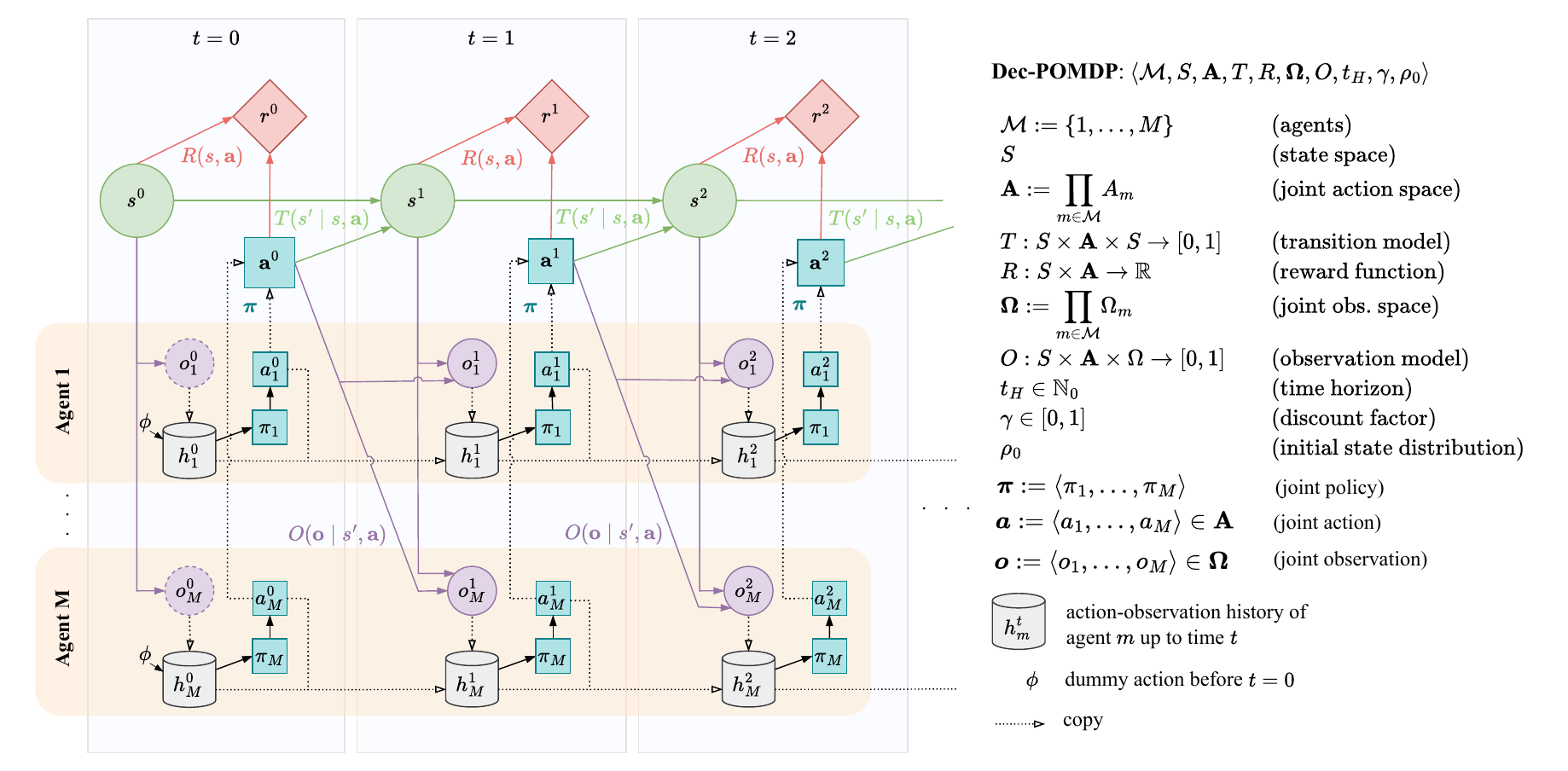}
    \caption{A Dec-POMDP unrolled as a dynamic decision network over three time steps. At each time step \(t\), the environment state \(s^t\) evolves according to the transition model \(s^{t+1} \sim T(\cdot \mid s^t, \mathbf{a}^t)\) under the joint action \(\mathbf{a}^t\), yielding a global reward $r^t$ and individual observations \(\mathbf{o}^{t+1} \sim O(\cdot \mid s^{t+1},\mathbf{a}^t)\) for each agent $m \in \mathcal{M}$. Each agent selects its action based on its local action–observation history $h_m^{t}$.}
    \label{fig:dec_pomdp}
\end{figure}

Dec-POMDPs extend the POMDP framework to multi-agent systems,  enabling scalability through decentralization~\citep{oliehoek_optimal_2008,oliehoek_concise_2016}. We illustrate the Dec-POMDP rolled out as a dynamic decision network over three time steps in Fig.~\ref{fig:dec_pomdp}. In this setting, agents act cooperatively, guided by a common reward signal to jointly maximize the expected return. It is defined by the tuple 
\(\langle 
\mathcal{M}, S, \mathbf{A}, T, R, \boldsymbol{\Omega}, O, t_H, \gamma, \rho_0
\rangle\), 
where 
\(\mathcal{M} \coloneqq \{ 1, \ldots, M \}\) is the set of agents, 
\(S\) is the state space, 
\(\textbf{A}\coloneqq \prod_{m \in \mathcal{M}} A_m\) is the joint action space, 
$T: S \times \textbf{A} \times S \to [0, 1]$ is the transition model, 
$R: S \times \textbf{A} \to \mathbb{R}$ is the reward model, 
$\boldsymbol{\Omega} \coloneqq \prod_{m \in \mathcal{M}} \Omega_m$ is the joint observation space, $O: S \times \textbf{A} \times \boldsymbol{\Omega} \to [0, 1]$ is the joint observation model, 
$t_H \in \mathbb{N}_0 \cup\{\infty\}$ is the time horizon,  
$\gamma \in [0, 1]$ is the discount factor, and \(\rho_0\) is the initial distribution over states.
To solve a Dec-POMDP, we seek an optimal \textit{joint} policy, $\boldsymbol{\pi}^*$, that maximizes the expected sum of discounted rewards:
\begin{equation}
\boldsymbol{\pi}^* \in \argmax_{\boldsymbol{\pi}} J(\boldsymbol{\pi}),  \quad
J(\boldsymbol{\pi}) \coloneqq
\mathbb{E} \left[
\sum_{t=0}^{t_H-1}\gamma^t R(s^t,\mathbf{a}^t)
\;\middle|\;
\begin{array}{l}
s^0 \sim \rho_0,\\
a_m^t \sim \pi_m(\cdot \mid h_m^t),\ \forall m\in\mathcal{M},\\
s^{t+1} \sim T(\cdot \mid s^t,\mathbf{a}^t),\\
\mathbf{o}^{t+1} \sim O(\cdot \mid s^{t+1},\mathbf{a}^t)
\end{array}
\right],
\label{eq:Dec-POMDP optimization}
\end{equation}
where the joint policy $\boldsymbol{\pi} \coloneqq \langle \pi_1, \ldots, \pi_M \rangle$ is a set of decentralized agent policies, and $\mathbf{a}^t \coloneqq \langle a_1^t, \ldots, a_M^t\rangle \in \textbf{A}$ is the joint action. The agent policy $\pi_m$ maps its individual action--observation history $h_m \in H_m$ to actions $a_m \in A_m$. It can be deterministic $a_m = \pi_m(h_m)$ or stochastic $a_m \sim \pi_m(h_m)$. Analogous to POMDPs, we initialize each agent's history at \(t=0\) with a null action, \(a_m^{-1} \coloneqq \phi \notin A_m\), and an initial observation, \(\mathbf{o}^0 \sim O^0(\cdot \mid s^0)\),
so that \(h_m^0 \coloneqq \langle a_m^{-1}, o_m^0 \rangle\).

\subsubsection{Multi-agent POMDP (MPOMDP)}

An MPOMDP is a special instance of a Dec-POMDP in which all agents always share their observations and past actions, allowing a common internal representation of joint action--observation histories $\mathbf{h} \in \mathbf{H}$ or joint beliefs $\mathbf{b} \coloneqq \langle \vec b_1, \ldots, \vec b_M \rangle \in \mathbf{B}$ between agents~\citep{oliehoek_concise_2016}. However, agents learn decentralized policies, mapping the shared representation to individual actions $\pi_m: \mathbf{H} \to A_m$. Since policies are factored agent-wise, this mitigates the exponential growth of the action space, but its scalability is limited by the dimensionality of $\mathbf{H}$.
\section{Multi-agent deep reinforcement learning (MADRL)}

Dec-POMDPs are hard to solve optimally as they require reasoning over the joint policy space that grows exponentially with the number of agents~\citep{oliehoek_concise_2016}.
To alleviate the curse of dimensionality, MADRL extends single-agent DRL (SADRL) by decentralizing decision-making across multiple agents. This extension is orthogonal to the native algorithmic classification in single-agent RL, such as model-based/model-free, value-based/actor--critic-based, or on/off-policy, and thus, many algorithmic variants that combine those dimensions exist. In this work, we focus on model-free approaches, including actor–critic and value-based methods, and cover both on-policy and off-policy algorithms.

\subsection{Paradigms and algorithms}

MADRL approaches aim to achieve scalability through decentralization of global information and/or joint actions. To alleviate the coordination difficulties arising from information decentralization, some algorithms assume access to centralized information, such as joint states/observations and joint actions, only during training but not during execution/deployment. A classification of algorithms based on information available during training and execution gives rise to three paradigms discussed below. An overview of algorithms corresponding to each paradigm studied in this work is shown in Fig.~\ref{fig:MARL_algorithms}, and a summary of corresponding formulation frameworks is presented in Table~\ref{tab:architectures}.

\begin{figure}
    \centering
    \includegraphics[width=1\linewidth]{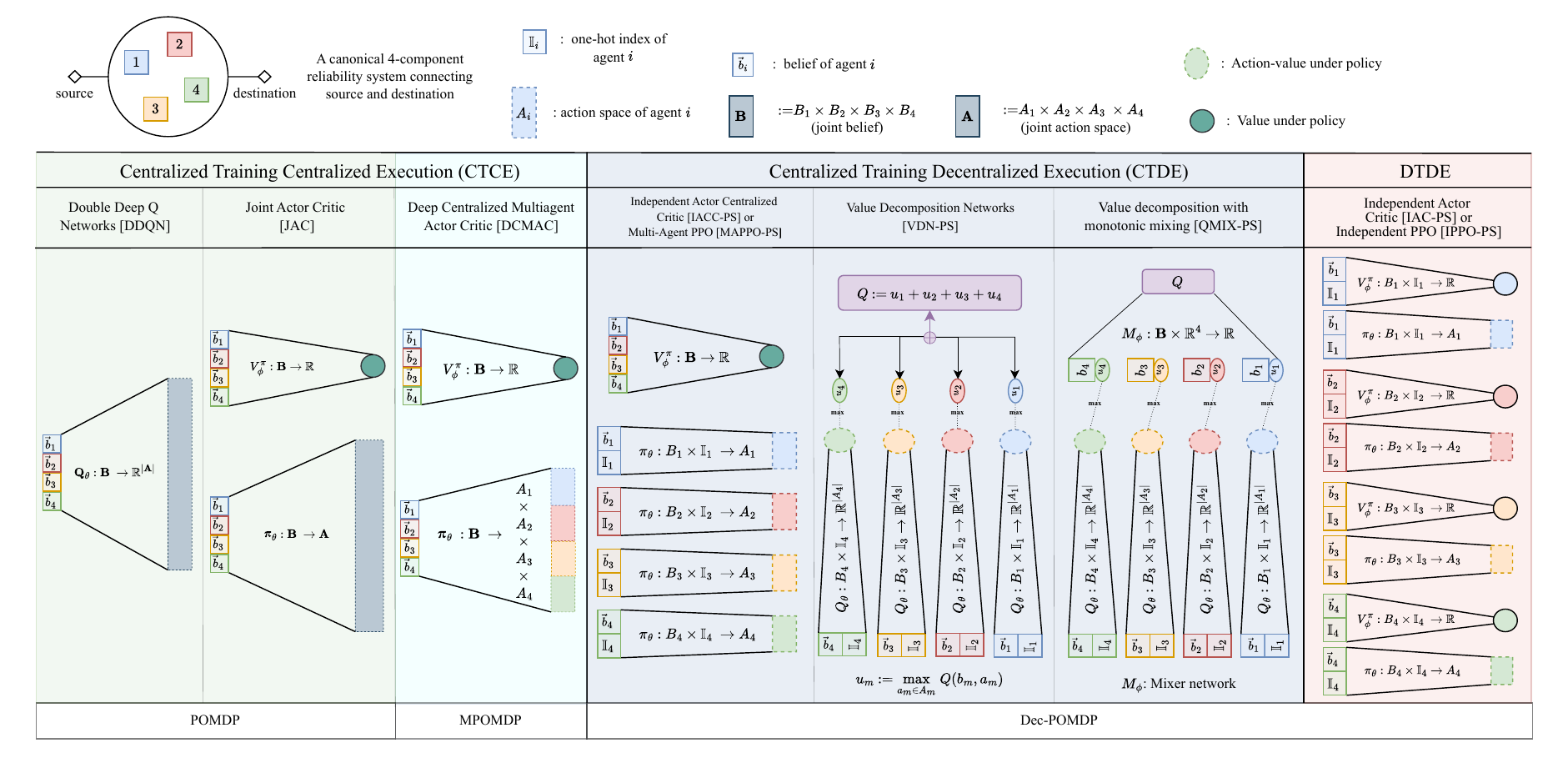}
    \caption{An overview of MADRL algorithms for managing a (four-component) reliability system. Algorithms are classified by training–execution paradigms: centralized training centralized execution (CTCE), centralized training decentralized execution (CTDE), and decentralized training decentralized execution (DTDE), reflecting increasing decentralization and information constraints. Bottom labels indicate the underlying formulation framework (POMDP, MPOMDP, Dec-POMDP). Suffix PS denotes parameter sharing across agents.}
    \label{fig:MARL_algorithms}
\end{figure}

\subsubsection{Centralized training centralized execution (CTCE)}

In this paradigm, agents assume access to centralized information during training and execution. A single, centralized agent learns a joint policy $\boldsymbol{\pi}$, mapping the joint action--observation history (or joint beliefs) to the joint/factorized action space, $\boldsymbol{\pi}: \mathbf{H} \to \mathbf{A}$. It becomes ``multi-agent'' only when the action space is factorized while the observations remain centralized (in an MPOMDP fashion).
The following algorithms convert the Dec-POMDP into a POMDP, and therefore inherit its scalability limitations.

\begin{itemize}
    \item Joint double DQN (DDQN), a single-agent value-based DRL algorithm as introduced in~\citep{van_hasselt_deep_2015}. It maps the joint action--observation history (or global beliefs) to the action--values of the discrete joint action space,\footnote{In practice, we typically learn an equivalent mapping $\mathbf{Q}_{\theta}: \mathbf{H} \to \mathbb{R}^{\lvert \mathbf{A}\lvert}$ (as shown in Fig.~\ref{fig:MARL_algorithms}).} $Q_{\theta}: \mathbf{H} \times \mathbf{A} \to \mathbb{R}$.

    \item Joint actor--critic (JAC), a single-agent actor--critic algorithm where the actor learns a stochastic policy $\boldsymbol{\pi}(\mathbf{a} \mid \mathbf{h}; \theta)$ mapping the joint action--observation history (or beliefs) to joint actions, and a centralized critic $V^{\boldsymbol{\pi}}(\mathbf{h}; \phi)$ learns the value function under that joint policy. The per-sample policy gradient $g_{\theta}$ for this algorithm is proportional to:
    \begin{equation}
        \label{eq:PG JAC}
        g_{\theta} \propto \nabla_{\theta} \mathrm{log}\, \pi(\mathbf{a} \mid \mathbf{h}; \theta).
    \end{equation} 
We deliberately write $\pi(\mathbf{a} \mid \mathbf{h}; \theta)$, instead of the usual $\pi_{\theta}(\mathbf{a} \mid \mathbf{h})$, to highlight the conditional dependence on parameters $\theta$, which will become important in the following algorithms.
\end{itemize}

The following algorithm converts the Dec-POMDP into an MPOMDP by relaxing the assumption of joint action spaces and exploiting the factored action space to improve scalability.
\begin{itemize}
    \item Deep centralized multi-agent actor--critic (DCMAC) learns a single parameter-shared actor that outputs a factorized joint policy, \(\boldsymbol{\pi}(\mathbf{a} \mid \mathbf{h}; \theta) = \prod_{m=1}^{M} \pi(a_m \mid \mathbf{h}; \theta)\), mapping the joint action--observation history \(\mathbf{h}\) to per-agent action distributions. Following Eq.~\eqref{eq:PG JAC}, the per-sample policy gradient for DCMAC algorithm is summarized as follows~\citep{sukthankar_cooperative_2017,andriotis_managing_2019}:
        \begin{equation}
        \label{eq:PG DCMAC}
        g_{\theta} \propto \nabla_{\theta} \Bigg[\sum_{m=1}^M \mathrm{log}\, \pi(a_m \mid \mathbf{h}; \theta) \Bigg].
    \end{equation}
The underlying assumption is that the actions are conditionally independent given $\mathbf{h}$ and shared parameters \(\theta\). Here, parameter sharing implies the same neural network parameters \(\theta\) are used for all agents, as illustrated in Fig.~\ref{fig:MARL_algorithms}.
Note that it is also possible to parameterize actors independently, as in deep decentralized multi-agent actor--critic (DDMAC)~\citep{andriotis_deep_2021}. This variant learns \(M\) separate actor policies \(\{\pi_m(a_m \mid \mathbf{h};\theta_m)\}_{m=1}^{M}\) that still condition on the joint action--observation history \(\mathbf{h}\), but use agent-specific parameters \(\{\theta_m\}\) instead of a shared \(\theta\).
\end{itemize}

\subsubsection{Centralized training decentralized execution (CTDE)} 

In this paradigm, algorithms learn to coordinate during training by being conditioned on centralized information, such as joint states, joint observations, joint actions, etc. However, during execution, each agent selects actions using only its local information.

\begin{itemize}
    \item Independent actor centralized critic (IACC), a multi-agent actor--critic algorithm where agents learn a stochastic policy $\pi_m(a_m \mid h_m; \theta_m)$ guided by a global/centralized critic $V^{\boldsymbol{\pi}}(\mathbf{h}; \phi)$ during training~\citep{lyu_centralized_2023}. Conceptually, IACC can be seen as the CTDE analog of DDMAC, and \citep{saifullah_multi-agent_2024} refers to IACC as DDMAC-CTDE. Recall that in DDMAC, each actor conditions on the joint history \(\mathbf{h}\). This joint-history input requires a neural network whose input dimension grows with the number of components, which becomes prohibitive at scale. IACC/DDMAC-CTDE alleviates this by conditioning each actor only on its local history \(h_m\). To further improve parameter efficiency and implicitly aid coordination, IACC can share policy parameters $\theta$ among actors, but use a one-hot encoded agent index $\mathbb{I}_m$ to enable the agents to learn distinct policies, $\pi_m = \pi(a_m \mid h_m, \mathbb{I}_m; \theta)$.\footnote{When agents have different number of actions, the one-hot vectors can be padded with zeros, and invalid action probabilities are set to 0~\citep{papoudakis_benchmarking_2021}.} We call this variant IACC-PS.
    \item Multi-agent proximal policy optimization (MAPPO) is similar to IACC-PS in that it uses decentralized stochastic policies trained with a centralized value function. Each agent learns a policy $\pi_m = \pi(a_m \mid h_m; \theta_m)$ via PPO updates using the clipped surrogate objective~\citep{yu_surprising_2022}:
    \begin{equation}
    \mathcal{L}^{\text{CLIP}}(\theta_m)
    =
    \mathbb{E}_{\tau \sim \boldsymbol{\pi}^{\text{old}}}\!\left[
    \min\!\left(
    r_m^t(\theta_m)\,\mathcal{A}^t,\ 
    \text{clip}(r_m^t(\theta_m), 1-\epsilon, 1+\epsilon)\,\mathcal{A}^t
    \right)
    \right],
    \label{eq:MAPPO_objective}
    \end{equation}

    where \(\tau=(s^0,\mathbf{a}^0,\mathbf{o}^1,s^1,\ldots)\) is a trajectory generated by the sampling process in Eq.~\eqref{eq:Dec-POMDP optimization} under the behavior policy \(\boldsymbol{\pi}^{\text{old}}\). \(\pi_m^{\text{old}}\) denotes the behavior policy used to collect trajectories, and  \(r^t(\theta_m) \coloneqq \frac{\pi_m(a^t_m \mid h^t_m; \theta_m)}{\pi^{\text{old}}_m(a^t_m \mid h^t_m; \theta_m)}\) is the importance weight that measures how the updated policy changes the probability of the sampled action relative to \(\pi_m^{\text{old}}\). The term \(\mathcal{A}^t\) is the generalized advantage estimate~\citep{schulman_proximal_2017}, computed using the centralized value function. Clipping stabilizes updates by preventing large policy shifts. Similar to IACC-PS, the parameter-sharing variant, MAPPO-PS, uses a single network conditioned on the agent index to learn distinct behaviors for each agent.

    \item Value decomposition networks (VDN) is a multi-agent value-based algorithm that represents the joint action--value $Q(\mathbf{h}, \mathbf{a})$, defined over the joint action–observation history $\mathbf{h}$ and joint action $\mathbf{a}$, as a sum of individual agent utilities $Q_m(h_m, a_m; \theta_m)$. The joint action--value is approximated as a linear combination of the agents' utilities~\citep{sunehag_value-decomposition_2017}:
    \begin{equation}
        Q(\mathbf{h}, \mathbf{a}) \approx \sum_{m=1}^M Q_m(h_m, a_m; \theta_m).
    \end{equation}
    VDN assumes that there exists a decomposition such that the joint greedy action can be recovered by independently maximizing each agent’s utility function, i.e., the individual-global-max (IGM) property holds:
    \begin{equation}
    \argmax_{\mathbf{a}\in \textbf{A}} Q(\mathbf{h}, \mathbf{a})
    =
    \left \langle
    \argmax_{a_1} Q_1(h_1, a_1),
    \dots,
    \argmax_{a_M} Q_M(h_M, a_M)
    \right \rangle.
    \end{equation}
    The parameter-sharing variant, VDN-PS, uses a single network shared across agents and conditioned on the agent index $\mathbb{I}_m$ to approximate individual utilities, i.e.,
    $Q_m(h_m, a_m) = Q(h_m, a_m \mid \mathbb{I}_m; \theta)$.

    \item Value decomposition with monotonic mixing (QMIX) extends VDN through monotonic mixing of agents' utilities, enabling it to represent a superset of centralized action--value function compared to VDN~\citep{rashid_qmix_2018,albrecht_multi-agent_2023}. QMIX uses the global state $\textbf{s}$ (or global beliefs ($\textbf{b}$)) along with the utilities to compute $Q(\mathbf{h}, \mathbf{a})$. The parameter-sharing variant, QMIX-PS, is similar to VDN-PS.
\end{itemize}

\subsubsection{Decentralized training decentralized execution (DTDE)}

\begin{itemize}
    \item Independent actor--critic (IAC) is a multi-agent actor--critic algorithm where agents learn a stochastic policy $\pi_m(a_m \mid h_m; \theta_m)$ guided by a local/decentralized critic $V^{\pi_m}(h_m; \phi_m)$ during training. In the parameter sharing variant (IAC-PS), the policies are modeled identically to IACC-PS, and the critics share parameters ($\phi$) as $V^{\pi_m} = V^{\pi}(h_m \mid \mathbb{I}_m; \phi)$ to represent the individual value functions.

    \item Independent proximal policy optimization (IPPO) is similar to IAC-PS in that each agent learns its own decentralized policy using local observations and rewards, without centralized training. Each agent independently optimizes a stochastic policy $\pi_m = \pi(a_m \mid h_m; \theta_m)$ using the clipped surrogate objective as defined in Eq.~\eqref{eq:MAPPO_objective} for MAPPO. However, unlike MAPPO, the advantage estimate for agent \(m\), \(\mathcal{A}^t_m\), is computed using its own local value function and does not use a centralized value function~\citep{dewitt2020independentlearningneedstarcraft}. As with IAC-PS, we can expedite learning by sharing parameters across actors and critics. We call this variant IPPO-PS.
\end{itemize}

\begin{table}[t]
\centering
\caption{Summary of single- and multi-agent DRL algorithms shown in Fig.~\ref{fig:MARL_algorithms}. The \emph{Actor} columns describe the policy, when present, with \emph{Input} indicating the information available to the policy and \emph{Output} indicating whether actions are produced jointly, \(\mathbf{a}\), or per agent, \(a_m\). The \emph{Critic/Value Function} columns report the information available to the value estimator during training and whether it outputs a centralized value (e.g., \(Q(\mathbf{h},\mathbf{a})\), \(V(\mathbf{h})\)) or per-agent values (e.g., \(Q_m(h_m,a_m)\)). \emph{Centralized} denotes access to global state or joint observations and actions, whereas \emph{Local} denotes access only to an agent’s own observation or action–observation history.}
\label{tab:architectures}
\resizebox{\linewidth}{!}{
\begin{tabular}{@{}cccccccccc@{}}
\toprule
\textbf{Paradigm} &
\textbf{Formulation} &
\multicolumn{3}{c}{\textbf{Algorithm}} &
\multicolumn{2}{c}{\textbf{Actor (actor--critic only)}} &
\multicolumn{2}{c}{\textbf{Critic/Value function}} \\
 &
\textbf{Framework} &
Name & Type & On/Off-policy &
Input & Output &
Input & Output \\
\midrule

\multirow{3}{*}{CTCE}
& \multirow{2}{*}{\begin{tabular}[c]{@{}c@{}}POMDP\\ (SADRL)\end{tabular}}
& DDQN & V & Off-policy
& -- & --
& Centralized & Centralized \\
& & JAC & A+C & Off-policy
& Centralized & Centralized
& Centralized & Centralized \\
\cmidrule(lr){2-5}
& MPOMDP & DCMAC & A+C & Off-policy
& Centralized & Local
& Centralized & Centralized \\

\midrule

\multirow{4}{*}{CTDE}
& \multirow{4}{*}{Dec-POMDP}
& IACC-PS & A+C & Off-policy
& Local & Local
& Centralized & Centralized \\
& & MAPPO-PS & A+C & On-policy
& Local & Local
& Centralized & Centralized \\
& & VDN-PS & V & Off-policy
& -- & --
& Centralized & Centralized \\
& & QMIX-PS & V & Off-policy
& -- & --
& Centralized & Centralized \\

\midrule

\multirow{2}{*}{DTDE}
& \multirow{2}{*}{Dec-POMDP}
& IAC-PS & A+C & Off-policy
& Local & Local
& Local & Local \\
& & IPPO-PS & A+C & On-policy
& Local & Local
& Local & Local \\

\bottomrule
\addlinespace[0.2em]
\end{tabular}
}
\parbox{\linewidth}{\footnotesize
\textbf{PS}: parameter sharing; \textbf{A+C}: actor--critic; \textbf{V}: value-based.
}
\end{table}

\subsection{Multi-agent learning pathologies}
\label{sec:multiagent_pathologies}
MADRL enables scalability through decentralized learning, but this decentralization can introduce cooperation pathologies that impede the learning of optimal policies. Key known pathologies are listed below ~\citep{matignon_independent_2012,claus_dynamics_1998}:
\begin{itemize}
    \item Non-stationarity problem: From the perspective of a decentralized agent, even a stationary environment appears non-stationary because other agents change their policies. Additionally, since agents explore independently, one agent's exploratory actions can perturb the others' learning signals, further encumbering learning.
    \item Pareto-selection problem: When multiple Pareto-optimal policies exist, agents must pick actions according to the same optimal policy because a mix of Pareto-optimal policies may not be optimal. For example, consider two optimal joint policies: $\mathbf{\pi_1^*} = \langle\textrm{left}, \textrm{left} \rangle$ and $\mathbf{\pi_2^*} = \langle\textrm{right}, \textrm{right}\rangle$. If agents select these policies independently, their actions can mismatch, producing $\mathbf{\pi} = \langle\textrm{left}, \textrm{right}\rangle$, which may not be optimal.
    \item Shadowed equilibria: An optimal equilibrium is said to be shadowed when a unilateral deviation from it yields a lower reward (or higher penalty) than a deviation from a suboptimal equilibrium. Due to the dependence on the actions of other agents, the optimal equilibrium becomes unattractive in favor of a safe suboptimal equilibrium~\citep{fulda_predicting_2007}. The reward 11 in the Climb Game in Fig.~\ref{fig:climb-game} is an example of a shadowed equilibrium.
\end{itemize}

\begin{figure}[t]
\begin{minipage}[c]{0.3\textwidth}
    \begingroup
\setlength{\tabcolsep}{8pt}
\renewcommand{\arraystretch}{1.2}
\begin{tabular}{@{}c c c c c@{}}
  &  & \multicolumn{3}{c}{Agent 1} \\
  & \multicolumn{1}{c|}{} & $\mathfrak{a}^1$ & $\mathfrak{a}^2$ & $\mathfrak{a}^3$ \\
  \cline{2-5}
  \multirow{3}{*}{\rotatebox{90}{\makebox[3em][c]{Agent 2}}}
    & \multicolumn{1}{c|}{$\mathfrak{a}^1$} & $11^{*}$ & $-30$ & $0$ \\
  & \multicolumn{1}{c|}{$\mathfrak{a}^2$} & $-30$ & $7$ & $6$ \\
  & \multicolumn{1}{c|}{$\mathfrak{a}^3$} & $0$ & $0$ & $5$ \\
\end{tabular}
\endgroup
\label{tab:climb_game}

\end{minipage}
\hfill
\begin{minipage}[c]{0.7\textwidth}
    \centering
    \includegraphics[width=\linewidth]{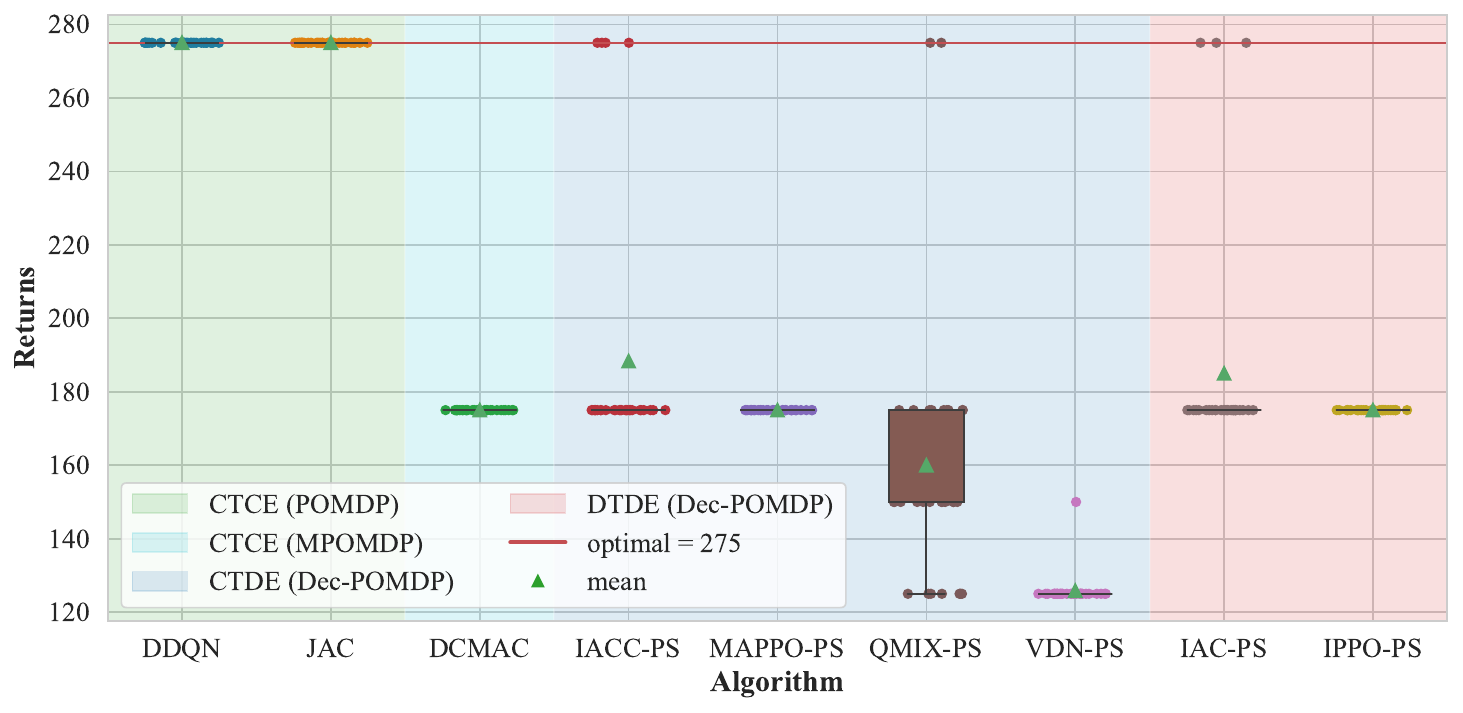}
\end{minipage}
    \caption{Left: Payoff matrix for the Climb Game with two agents, each choosing from three actions. The joint action \((\mathfrak{a}^1, \mathfrak{a}^1)\) yields the optimal reward of 11~\citep{claus_dynamics_1998}.  Right: Best performance of MADRL algorithms aggregated over 30 training instances in the Climb Game played repeatedly for 25 time steps~\citep{papoudakis_benchmarking_2021}. The optimal return is $11\times25 = 275$.}
    \label{fig:climb-game}
\end{figure}

\paragraph{Climb Game} A canonical example that illustrates the shadowed equilibria pathology is the Climb Game introduced by~\cite{claus_dynamics_1998}. It involves two agents, each with three actions, receiving rewards based solely on their joint action, as shown in the payoff matrix in Fig.~\ref{fig:climb-game}. The agents do not receive any observations and must learn to disentangle the rewards through repeated interaction. The optimal reward (11) is shadowed by penalties ($-30$) when an agent unilaterally deviates from the optimum.
\section{Multi-agent I\&M environments}
\label{sec:I&M environments}

\subsection{Problem statement}
As introduced in Section~\ref{sec:background_dec_pomdp}, the Dec-POMDP framework formalizes decentralized sequential decision-making under uncertainty. In the context of I\&M planning, it provides a structured way to model multi-component or multi-subsystem engineering systems, where decentralized agents must act based on local observations and coordinate to manage deterioration and minimize long-term inspection/maintenance costs. It encapsulates a deteriorating engineering system where:
 \begin{itemize}
     \item $\mathcal{M} \coloneqq \{1, \ldots, M\}$ is the set of agents acting on components (or subsystems);
     \item $\mathbf{S}\coloneqq\times_{m \in \mathcal{M}} S_m \times S_g$  is the discrete/continuous joint state space. It is often an aggregation of deterioration properties $S_m$  of the components, such as damage state, deterioration rates, etc., and global features $S_g$, such as time, resource availability, etc.;
     \item $\mathbf{A} \coloneqq \times_{m \in \mathcal{M}} A_m$ denotes the joint action space, with $A_m$ representing the set of component-level actions available to an agent $m$ (e.g., do-nothing, inspection, repair, replacement, or decommissioning);
     \item \(T: \mathbf{S} \times \mathbf{A} \times \mathbf{S} \to [0,1]\) denotes the Markovian transition model summarizing the dynamics of system deterioration and the effect of the intervention actions~\citep{van_noortwijk_survey_2009,molaioni2024dynamic}. If components deteriorate independently, $T(\mathbf{s}, \mathbf{a}, \mathbf{s'})$ can simply be decomposed component-wise, and when deterioration is correlated or dependent, it can be decomposed into factorized representations by reconstructing the underlying dynamic Bayesian network~\citep{morato_inference_2023};
     \item \(C: \mathbf{S} \times \mathbf{A} \to \mathbb{R}\) denotes the cost model summarizing the costs associated with joint intervention actions, risk, revenue, system failure, etc.;
     \item \(\boldsymbol{\Omega}\coloneqq\times_{m \in \mathcal{M}} \Omega_m\) denotes the set of joint observations representing the signals we obtain from inspections and monitoring systems about the joint state of the system. It could be as simple as indicating whether the system is functional. Similar to system states, it can also be decomposed component-wise when observations are independent;
     \item \(O: \boldsymbol{\Omega} \times \mathbf{A} \times \mathbf{S} \to [0,1]\) is the observation model summarizing the accuracy of observations under each intervention action;
     \item $t_H \in \mathbb{N}_0 \cup \{\infty\}$ is the (in)finite horizon over which the system must be maintained, where $t_H = \infty$ denotes the infinite-horizon case;
     \item $\gamma \in [0, 1]$ specifies the trade-off between prioritizing current vs future costs;
     \item $\rho_0$ is the initial distribution over states.
 \end{itemize}
The goal is to find an optimal joint policy that minimizes the expected sum of costs over the given horizon. We rewrite Eq.~\eqref{eq:Dec-POMDP optimization} as a minimization problem as follows:
\begin{equation}
    \boldsymbol{\pi}^* \in \argmin_{\boldsymbol{\pi}} J(\boldsymbol{\pi}),
    \quad
    J(\boldsymbol{\pi}) \coloneqq
    \mathbb{E}\left[
    \sum_{t=0}^{t_H-1}\gamma^t C(s^t,\mathbf{a}^t)
    \;\middle|\;
    \begin{array}{l}
    s^0 \sim \rho_0,\\
    a_m^t \sim \pi_m(\cdot \mid h_m^t),\ \forall m\in\mathcal{M},\\
    s^{t+1} \sim T(\cdot \mid s^t,\mathbf{a}^t),\\
    \mathbf{o}^{t+1} \sim O(\cdot \mid s^{t+1},\mathbf{a}^t)
    \end{array}
    \right].
   \label{eq:I&M optimization}
\end{equation}

In a general Dec-POMDP formulation, each policy can be conditioned on the agent's action--observation history \(h_m^t\). In many I\&M applications, the deterioration transition and observation models are available, allowing each agent to recursively convert its local \(h_m^t\) into a belief state using the Bayesian update in Eq.~\eqref{eq:belief_update}. The resulting local belief,
\(
b_m^t(s_m)=\Pr\!\left(s_m^t=s_m \mid h_m^t\right),
\)
represents the agent's uncertainty about the condition of component \(m\). Since the belief is a sufficient statistic for the action--observation history, the policy can be expressed as \(\pi_m(a_m^t\mid b_m^t)\) without discarding decision-relevant information~\citep{astrom_optimal_1965,kaelbling_partially_1995}. Belief-conditioned policies are widely used in practical I\&M applications~\citep{Bhattacharya2025,saifullah_multi-agent_2024,arcieri_bridging_2023}.

\subsection{The price of decentralization}
\label{section:price_of_decentralization}
Decentralized learning improves scalability by distributing decision-making across agents, but this also introduces the multi-agent learning pathologies discussed in Section~\ref{sec:multiagent_pathologies}. These pathologies can impede coordination and prevent convergence to optimal policies, creating a trade-off between scalability and solution optimality. To quantify this trade-off, we define the \emph{price of decentralization} ($\Delta_{\text{dec}}$) as the relative optimality loss incurred by decentralized control with respect to the centralized optimum:
\begin{equation}
\Delta_{\text{dec}} \coloneqq \frac{C_{\text{dec}}-C^*}{C^*} \times 100 \%,
\label{eq:decentralization_gap}
\end{equation}
where $C_{\text{dec}} \coloneqq J(\boldsymbol\pi)$ is the expected cost attained by a decentralized policy and $C^* \coloneqq J(\boldsymbol\pi^*)$ is that of the optimal centralized policy, both evaluated using the cost-minimization objective $J$ for I\&M planning in Eq.~\eqref{eq:I&M optimization}. Since $C^*$ is the minimum achievable cost, $\Delta_{\text{dec}}\geq0$, with equality holding when the decentralized policy attains the centralized optimum. In practice, $\Delta_{\text{dec}}$ can arise from both the restrictions on information and control imposed by decentralization and function-approximation and optimization errors of the learning algorithm.

Quantifying $\Delta_{\text{dec}}$ requires access to $C^*$, which is computationally tractable only for small systems. This motivates the benchmark environments introduced in Section~\ref{sec:k-out-of-n infinite-horizon environments}. For these environments, $C^*$ can be computed using POMDP solvers such as SARSOP, which provide bounds on solution quality and thus optimality guarantees. The estimation of $C_{\text{dec}}$ and $C^*$ follows the evaluation protocol described in Section~\ref{section:experimental_setup}.

\subsection{\texorpdfstring{\koutofn{}}{k-out-of-n} infinite-horizon environments}
\label{sec:k-out-of-n infinite-horizon environments}
To study the effect of redundancy on the price of decentralization in multi-agent decision-making, we require environments with realistic system dynamics and cost models. In I\&M planning, cost models reflect long-term costs and delayed outcomes, which influence learning and the emergence of multi-agent pathologies. The \koutofn{} system allows us to model the core features of deterioration and maintenance, and vary redundancy in a systematic manner.

A \koutofn~system, a common archetype of deteriorating engineering systems, operates as long as at least \(k\) out of its \(n\) components remain functional. This system includes two special cases: \(k=n\), which corresponds to a series system where all components must function for the system to operate, and \(k=1\), representing a parallel system where the operation of any single component is sufficient for the system to function \citep{kapur_reliability_2014, kuo_optimal_2003}. 
We focus on a system with \(n = 4\) components, where each component is managed by a dedicated agent, resulting in \( M = n \). 
Varying the redundancy parameter \(k\) from 1 (parallel) to \(n\) (series) yields the family of \koutofn{} environments studied in this work. 
The \(n=4\) system serves as the base case, as it is the largest configuration supported by SARSOP, to obtain (near-)optimal baselines, while still capturing the essential features of deterioration and maintenance. 
Systems with \(n \in \{2,3\}\), defined analogously as subsets of the base case, are also studied and reported in~\ref{appendix:ablation_study} alongside ablation studies to verify consistency across different system sizes and reward specifications. 
In the following, we define all component-level parameters (transition models, costs, and observation models) for the \(n=4\) system; the smaller systems use consistent subsets of these definitions.

\paragraph{State and action spaces}  The state space of each component $S_m$ represents its possible damage states and is defined as 
$S_m \coloneqq \{\mathtt{s^1} = \text{no-damage}, \mathtt{s^2} = \text{major-damage}, \mathtt{s^3} = \text{failure}\}.$ 
Similarly, the actions available for each component are given by $A_m \coloneqq \{\mathtt{a^1} = \text{do-nothing}, \mathtt{a^2} = \text{repair}, \mathtt{a^3} = \text{inspect}\}$. 

\paragraph{Transition models} Components have unique, independent, and stationary deterioration models \newline 
\(T(s_m' \mid s_m, \textrm{do-nothing/inspect}) = T_m^\text{do-nothing/inspect}\). The synthetic transition models described below reflect the key features of a deterioration model: they are upper-triangular (indicating that a component's state cannot improve without repair) and have unit probability at the terminal state (signifying that failure is an absorbing state). The transition models for the do-nothing and inspect actions are:
\[
\begin{array}{cccc}
\overset{T_1^\text{do-nothing/inspect}}{\begin{bmatrix}
0.82 & 0.18 & 0 \\
0 & 0.87 & 0.13 \\
0 & 0 & 1
\end{bmatrix}} & 
\overset{T_2^\text{do-nothing/inspect}}{\begin{bmatrix}
0.72 & 0.28 & 0 \\
0 & 0.78 & 0.22 \\
0 & 0 & 1
\end{bmatrix}} & 
\overset{T_3^\text{do-nothing/inspect}}{\begin{bmatrix}
0.79 & 0.21 & 0 \\
0 & 0.85 & 0.15 \\
0 & 0 & 1
\end{bmatrix}} & 
\overset{T_4^\text{do-nothing/inspect}}{\begin{bmatrix}
0.72 & 0.28 & 0 \\
0 & 0.78 & 0.22 \\
0 & 0 & 1
\end{bmatrix}}.
\end{array}
\]
Components can be restored to the \(\textrm{no-damage}\) state by the repair action, 
which succeeds with probabilities \newline \(\{0.95, 0.9, 0.98, 1\}\) for the four components, respectively. The transition model for the repair action for each component \(T_m^r(s_m' \mid s_m, \textrm{repair})\) is defined as:
\[
\begin{array}{cccc}
\overset{T_1^\text{repair}}{\begin{bmatrix}
1 & 0 & 0 \\
0.95 & 0.05 & 0 \\
0.95 & 0 & 0.05
\end{bmatrix}} & 
\overset{T_2^\text{repair}}{\begin{bmatrix}
1 & 0 & 0 \\
0.9 & 0.1 & 0 \\
0.9 & 0 & 0.1
\end{bmatrix}} & 
\overset{T_3^\text{repair}}{\begin{bmatrix}
1 & 0 & 0 \\
0.98 & 0.02 & 0 \\
0.98 & 0 & 0.02
\end{bmatrix}} & 
\overset{T_4^\text{repair}}{\begin{bmatrix}
1 & 0 & 0 \\
1 & 0 & 0 \\
1 & 0 & 0
\end{bmatrix}}.
\end{array}
\]
The transition model for each component can be summarized as follows:
\begin{equation}
T_m \coloneqq 
\begin{cases} 
T_m^\text{do-nothing/inspect}, & \text{if } a = \text{do-nothing or inspect}, \\
T_m^\text{repair} \times T_m^\text{do-nothing/inspect}, & \text{if } a = \text{repair}.
\end{cases}
\end{equation}

\paragraph{Cost models} The cost model consists of two parts: component-level costs (repair and inspection), and system-level costs (risk and mobilization). Repair and inspection costs are unique for each component and are given by: $\{50, 30, 80, 90\}$ and $\{ 5, 3, 8, 4\}$, respectively. 
In the \(n=4\) system, taking actions incurs a mobilization cost \(c_{{\mathrm{{mob}}}} = 4\), which incentivizes independent agents to coordinate the timing of interventions. Results for variations without mobilization costs (with \(n \in \{2,3,4\}\)) are reported in~\ref{appendix:ablation_study}. Finally, the system-level risk is defined as \(c_f \times p_f\), where the cost of failure is 
\(
c_f = \sum_{m=1}^n \texttt{repair\_reward}_m \times \kappa.
\)
We fix the failure penalty factor, \(\kappa = 3\), while \(\sum_{m=1}^n \texttt{repair\_reward}_m\) depends on the number of components \(n\). The probability of system failure \(p_f\) is computed using the closed-form expression for \koutofn{} systems~\citep{kuo_optimal_2003}.

\paragraph{Observation models} Every agent has a unique observation accuracy with which it perceives the damage state of the component during inspection. The observation accuracies for each component are \(0.8, 0.85, 0.9, 0.9\), respectively. Furthermore, the failure state is not fully observable, and has failure observation accuracies of \(0.95, 0.9, 0.98, 1\), respectively. The observation model for inspection for each component \(O_m^\text{inspect} \coloneqq O_m(o_m \mid s_m, \textrm{inspection})\), and no inspection \(O_m^\text{do-nothing/repair} \coloneqq O_m(o_m \mid s_m, \{\textrm{do-nothing, repair}\})\) are:
\[
\begin{array}{cccc}
\overset{O_1^\text{inspect}}{\begin{bmatrix}
0.8 & 0.2 & 0 \\
0.1 & 0.8 & 0.1 \\
0 & 0.05 & 0.95
\end{bmatrix}}
\overset{O_2^\text{inspect}}{\begin{bmatrix}
0.85 & 0.15 & 0 \\
0.075 & 0.85 & 0.075 \\
0 & 0.1 & 0.9
\end{bmatrix}}
\overset{O_3^\text{inspect}}{\begin{bmatrix}
0.9 & 0.1 & 0 \\
0.05 & 0.9 & 0.05 \\
0 & 0.02 & 0.98
\end{bmatrix}}
\overset{O_4^\text{inspect}}{\begin{bmatrix}
0.9 & 0.1 & 0 \\
0.05 & 0.9 & 0.05 \\
0 & 0 & 1
\end{bmatrix}}.
\end{array}
\]
Without inspection (do-nothing or repair), observations are uninformative: \(O_m^\text{do-nothing/repair}\) has all entries equal to \(1/3\) for every \(m\).

\paragraph{Initial belief, time horizon, and discount factor} All components have the same initial belief, $b_0 = (0.6, 0.4, 0)$ and the system must be maintained for an infinite number of time steps, using the infinite-horizon objective in Eq.~\eqref{eq:I&M optimization} with a discount factor \(\gamma = 0.8\). This infinite-horizon formulation avoids modeling time in the state, preventing linear growth in the state space and allowing us to scale to multiple components using the point-based solver, SARSOP, to obtain (near-)optimal baselines. During training, each episode is truncated at 50 time steps to stabilize learning under infinite-horizon discounting, as the returns become negligible \((0.8^{50} \approx 10^{-5})\). Since the policy is stationary, inference is limited to 20 time steps to reduce computational cost. Performance is evaluated using the discounted sum of rewards.
\section{Experimental setup}
\label{section:experimental_setup}
We examine how redundancy impacts coordination and learning by benchmarking MADRL algorithms spanning CTCE, CTDE, and DTDE on the Climb Game and on the \(k\)-out-of-\(n\) environments described in the previous sections. Performance is measured by the expected discounted return (or cost) \(J\) from Eq.~\eqref{eq:I&M optimization}, together with variability across random seeds. To interpret the results, we compare against (near-)optimal SARSOP policies and an optimized heuristic baseline, noting that the Monte Carlo estimate of \(J\) for a given decentralized policy yields \(C_{\text{dec}}\) and \(C^*\) for the (near-)optimal SARSOP policy. To validate these findings robustly, we perform ablations that remove mobilization costs and vary the failure penalty factor \(\kappa\), confirming that the trends persist across cost specifications.

\subsection{Evaluation protocol}
We follow the evaluation protocol outlined in~\citep{papoudakis_benchmarking_2021,gorsane_towards_2022}. Each algorithm is trained for a fixed computational budget and evaluated periodically using Monte Carlo rollouts on the objective \(J\) in Eq.~\eqref{eq:Dec-POMDP optimization} for the Climb Game and Eq.~\eqref{eq:I&M optimization} for \koutofn{} systems. All algorithms are trained for 2.5M time steps in the Climb Game. In the \koutofn{} systems, off-policy algorithms are trained for 2.5M time steps, while on-policy algorithms (MAPPO-PS, IPPO-PS) are trained for 20M time steps. The larger interaction budget for on-policy methods reflects their need to collect new environment interactions for every policy update, whereas off-policy methods can perform multiple updates using previously collected transitions stored in a replay buffer.

Because our objective is to compare the best policy each algorithm can attain within the fixed training budget, we represent each run by its best-performing checkpoint. Since training performance is non-monotonic and the best policy may occur at any point, we evaluate each run at 11 uniformly spaced checkpoints. Each checkpoint is evaluated with 100{,}000 Monte Carlo rollouts due to the high variance of \(J\)~\citep{leroy_imp-marl_2023}. During evaluation, exploration is disabled (epsilon set to zero for epsilon-greedy policies), while stochastic policies sample actions. Furthermore, the consistency of the algorithm is measured by training several instances of each algorithm by varying neural network parameter initializations and environment conditions using random seeds.

\subsection{Environments, baselines, and learning details}

\paragraph{Climb Game} Following~\cite{papoudakis_benchmarking_2021}, we set up a repeated matrix game over 25 time steps to assess the performance of the various MADRL algorithms. Since the optimal reward is \(11\), the optimal full-horizon return is \(11\times25=275\). We provide a constant input of zero to the neural networks, and in the case of multi-agent algorithms with parameter sharing, the neural network is additionally provided with agent IDs as input. We assess the performance of all algorithms over 30 random seeds.

\paragraph{\koutofn{} environments} During training, each episode is truncated after 50 time steps, which practically suffices to approximate an infinite-horizon problem as explained in Section \ref{sec:I&M environments}.
Since truncation is due to a time limit (and not due to an absorbing terminal state), the value at the final step is in general nonzero. Accordingly, when computing returns and advantages, we bootstrap at the truncation boundary using the critic. This is in contrast to true terminations, such as in finite-horizon settings, where the bootstrap value is set to zero.
During inference, we evaluate the performance over 20 time steps using 100,000 Monte Carlo rollouts. During both training and inference, we report the expected discounted returns. Performance is assessed on 10 random seeds for each algorithm. In these environments, agents only receive beliefs over component damage states as input. Time steps are omitted from the input since policies are stationary in infinite-horizon environments. For multi-agent algorithms with parameter sharing, agent IDs are additionally provided. The hyperparameters for all algorithms are tuned on the 3-out-of-4 system, and are listed in Table~\ref{tab:hyperparameters}.

\paragraph{SARSOP baseline} To evaluate the performance of MADRL algorithms in the infinite-horizon environments, we establish a (near-)optimal baseline using SARSOP~\citep{kurniawati_sarsop_2008}, an offline planning algorithm to solve POMDPs. The SARSOP policy is further improved using one-step look-ahead (see Section \ref{sec:POMDP solvers} for details). We evaluate the performance using 100,000 Monte Carlo rollouts, and also use it to contrast DRL policies in belief-space plots. SARSOP also provides upper and lower bounds on expected returns, which we report along with the Monte Carlo returns. In this work, we used an existing \texttt{Julia} implementation of SARSOP, with POMDPs specified in the \texttt{POMDPs.jl} framework~\citep{Julia_POMDPs}. This allowed us to model the \koutofn{} systems directly in \texttt{Julia} and obtain (near-)optimal policies while maintaining consistency with the environments used to train the DRL algorithms. Convergence plots for SARSOP for the \(k\)-out-of-4 systems are shown in Fig.~\ref{fig:sarsop_convergence}.

\paragraph{Heuristic baseline} We evaluate an inspection–repair heuristic defined by three parameters: (i) the inspection interval, which sets the periodicity of inspections after which repair actions may be triggered, (ii) the number of components inspected at each interval, with priority given to those with the highest estimated failure probabilities, and (iii) the repair threshold, which specifies the minimum observed damage level at inspection that triggers repair. We search each parameter over its admissible range: inspection intervals from 1 up to 20, number of inspected components from 1 up to \(n\), and repair thresholds across the number of damage states. For each parameter combination, we compute the discounted return using 100,000 Monte Carlo rollouts and report the best-performing heuristic as a baseline for comparison.

\subsection{Ablations and validation}
\label{section:ablations_and_validation}

We additionally study \(n=2\) and \(n=3\) systems, defined as consistent subsets of the \(n=4\) system by reusing the first two or three components, respectively (including their deterioration, cost, and observation parameters). These systems exclude mobilization costs to isolate the effect of redundancy on coordination. We also conduct an ablation with an alternate reward model, where the failure penalty factor is reduced from \(\kappa = 3\) to \(\kappa = 1.5\). These additional results, reported in Fig.~\ref{fig:ablation_study}, show that the main findings hold across different system sizes and reward specifications, confirming that the observed trends are not artifacts of a particular cost structure. Furthermore, we validate our MADRL implementations by comparing a subset of algorithms, QMIX-PS, MAPPO-PS, and IPPO-PS, against the open-source MARL library EPyMARL~\citep{papoudakis_benchmarking_2021} (see~\ref{appendix:validation} for details).

Our benchmark environments, algorithm implementations, baselines, and trained models for reproducing the results of this paper are publicly available at
\href{https://github.com/prateekbhustali/price-of-decentralization}{\texttt{github.com/prateekbhustali/price-of-decentralization}}.

\section{Results and discussion}
\label{sec:results}

\begin{figure}[t]
    \centering
    \includegraphics[width=1\linewidth]{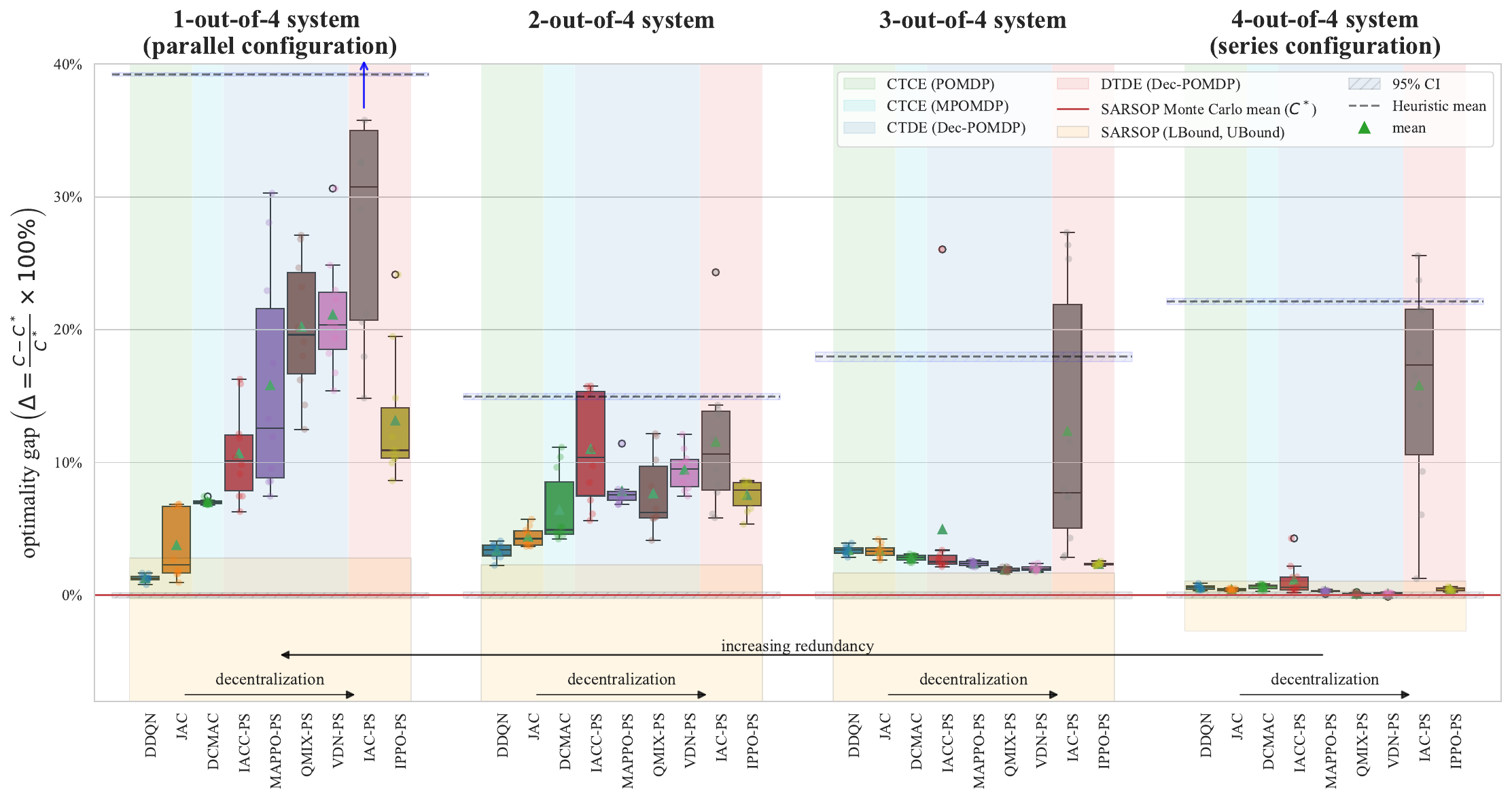}
    \caption{Box plots summarize algorithm performance across \(k\)-out-of-4 systems, aggregated over 10 training instances. To enable comparison across varying redundancy levels, performance is reported as the optimality gap, $\frac{C-C^*}{C^*}\times100\%$, where $C$ is the expected lifecycle cost of the evaluated policy and $C^*$ is the expected cost of the (near-)optimal SARSOP policy (Table~\ref{tab:best_results_kn_infinitehorizon}). Blue arrow ({\color{blue}$\uparrow$}) indicates entries that exceed the plotted \(y\)-axis limits. Most decentralized algorithms perform near-optimally in series configurations \((k=n)\), but deviate from the (near-)optimal SARSOP baseline as redundancy increases, with the largest optimality gaps observed in the parallel configurations \((k=1)\).}
    \label{fig:kn_infinite_horizon}
\end{figure}

\begin{figure}[t]
    \centering
    \includegraphics[width=0.7\linewidth]{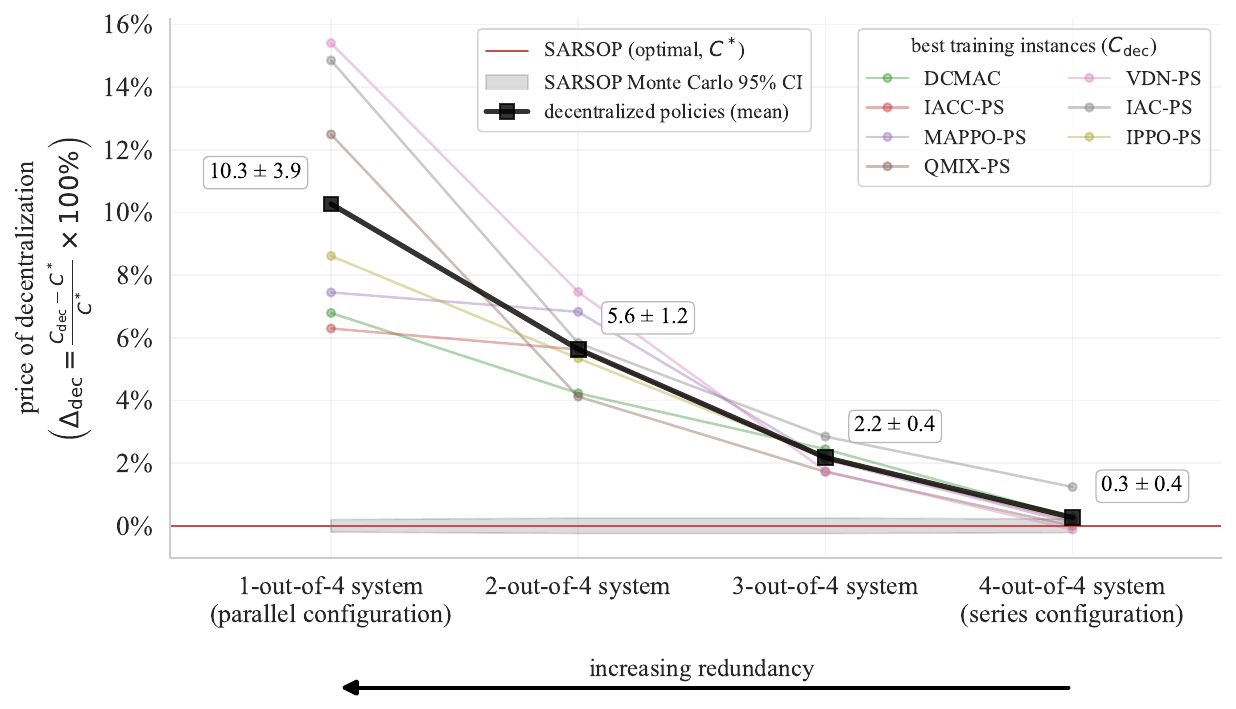}
    \caption{Price of decentralization across redundancy levels. The price of decentralization measures the relative optimality loss $\left(=\frac{C_{\text{dec}}-C^*}{C^*}\right)$ of each decentralized training instance over the optimal centralized baseline (approximated here by SARSOP). We compare the best training instance of each algorithm (thin lines) and the trend across decentralized algorithms (thick line) across the four redundancy levels; annotations report the mean $\pm$ one standard deviation across the thin lines. The price grows monotonically with system redundancy from the 4-out-of-4 (series) system with no redundancy to the 1-out-of-4 (parallel) system.}
    \label{fig:pod_vs_redundancy}
\end{figure}

Following the evaluation protocol outlined in the previous section, the box plots in Figs.~\ref{fig:climb-game} and \ref{fig:kn_infinite_horizon} show the best performance over 30 (Climb Game) and 10 (\(k\)-out-of-4 systems) training runs (random seeds) for each DRL algorithm. In Fig.~\ref{fig:kn_infinite_horizon}, we normalize results with respect to the (near-)optimal SARSOP baseline to enable comparison across system configurations. Fig.~\ref{fig:pod_vs_redundancy} summarizes this comparison in terms of the price of decentralization across redundancy levels. The mean optimality gap across decentralized algorithms increases monotonically from $0.3\% \pm 0.4\%$ in the 4-out-of-4 series system to $10.3\% \pm 3.9\%$ in the 1-out-of-4 system. This provides quantitative evidence that decentralized performance degrades as system redundancy increases.
In Table~\ref{tab:best_results_kn_infinitehorizon}, the LBound is an approximate lower bound on costs, and UBound is the upper bound on costs computed through the point-based backups.

\paragraph{Learning curves} Learning curves are plotted using the 11 uniformly spaced evaluations performed during training. Each point reports the average return or cost at the corresponding checkpoint. Figs.~\ref{fig:learning_curves_climb_game} and~\ref{fig:learning_curves_kn_infinite} show the resulting learning curves for the Climb Game and all \(k\)-out-of-4 environments, respectively. 

\paragraph{Belief space plots} In Fig.~\ref{fig:belief_occupany_1of4}, we analyze the best-performing policies using the belief space for the 1-out-of-4 system. They serve as a lens for contrasting the (near-)optimal SARSOP baseline, centralized DRL policies, and decentralized MADRL policies. In these plots, each 2D simplex represents the belief over 3 states, expressed in barycentric coordinates, where the vertices correspond to the damage states $\{\mathtt{s_1} = \text{no-damage}, \mathtt{s_2} = \text{major-damage}, \mathtt{s_3} = \text{failure}\}$. At each vertex, the damage state is perfectly known; uncertainty increases as we move towards the interior and is maximum at the centroid. On the belief space, we plot the actions taken by the agent(s) over 100 Monte Carlo rollouts starting from the given initial belief $b_0$. If no repair actions are taken, components slowly drift towards the failed state (\(\mathtt{s_3} = \text{failure}\)).

\paragraph{Cost decomposition} Fig.~\ref{fig:excess_cost_vs_sarsop} in \ref{appendix:additional_results} decomposes each algorithm's excess cost relative to SARSOP into mobilization, repair, inspection, and risk costs. No single cost term consistently explains the gap across algorithms or redundancy levels. Instead, suboptimality manifests differently across learned policies, including over-repair and excess mobilization costs, misaligned inspection or intervention cycles, or elevated risk from poorly captured component dependencies.

\subsection{CTCE algorithms}

\subsubsection*{JAC and DDQN}
The SADRL approaches, JAC and DDQN, learn centralized joint policies and consistently achieve the optimal return in the Climb Game, as shown in Fig.~\ref{fig:climb-game}. Their performance remains near-optimal across most \(k\)-out-of-\(4\) systems, as shown in Fig.~\ref{fig:kn_infinite_horizon} and Table~\ref{tab:best_results_kn_infinitehorizon}. 
Notably, JAC exhibits suboptimal behavior in the 1-out-of-4 system, evident from the bimodal distribution in Fig.~\ref{fig:kn_infinite_horizon}. To substantiate this further, we report performance across all 10 random seeds in Table~\ref{tab:all_results} in \ref{appendix:additional_results}. 
In the 1-out-of-4 system, the worst-performing JAC instance (seed 10), the learned joint policy maintains only component 2, while SARSOP and optimal JAC policies maintain both components 1 and 2. This is visualized in a policy rollout in Fig.~\ref{fig:suboptimal_JAC} in~\ref{appendix:additional_results}. In contrast, DDQN exhibits greater stability in performance and is also more parameter-efficient to train, as reported in Table~\ref{tab:hyperparameters}.

\subsubsection*{DCMAC}
In the Climb Game, DCMAC consistently converges to a suboptimal policy that selects the reward-7 action, yielding a return of \(25 \times 7 = 175\). One possible explanation is DCMAC's factorization of the joint policy \(\boldsymbol{\pi}(\mathbf{a} \mid \mathbf{h}) = \prod_{m=1}^M \pi_m(a_m \mid \mathbf{h})\) (Eq.~\eqref{eq:PG DCMAC}) to address the curse of dimensionality. Following Eq.~\eqref{eq:PG JAC}, the per-sample joint policy gradient update for JAC in the Climb Game is:
\begin{equation}
g_\theta = \nabla_\theta \Big[\log
\underbrace{\pi(a_2 \mid a_1;\theta)}_{\text{conditioned on $a_1$}}
+ \log \pi(a_1;\theta)\Big]\cdot\mathcal{A}, 
\qquad \textrm{(JAC)}
\label{eq:PG JAC Climb}
\end{equation}
where $\mathcal{A}$ is the joint advantage. Since there are no observations or states, \(\mathbf{h}\) is omitted. In contrast, DCMAC assumes conditional independence of actions given the policy parameters, i.e., \(\pi(a_2 \mid a_1; \theta) = \pi(a_2; \theta)\), and thus, following Eq.~\eqref{eq:PG DCMAC}:
\begin{equation}
g_\theta = \nabla_\theta \Big[\log 
\underbrace{\pi_2(a_2;\theta)}_{\text{independent of $a_1$}}
+ \log \pi_1(a_1;\theta)\Big]\cdot\mathcal{A}, 
\qquad \textrm{(DCMAC)}
\label{eq:PG DCMAC Climb}
\end{equation}

\noindent which updates each agent’s action probability independently with the same advantage estimate. We posit that this conditional independence assumption contributes to the observed performance gap.

Despite these limitations, DCMAC generally performs well in the series and series-like configurations in the \(k\)-out-of-4 systems. However, as shown in Fig.~\ref{fig:pod_vs_redundancy}, performance tends to become suboptimal in the presence of redundancies, suggesting that access to centralized information alone is insufficient to fully resolve coordination challenges induced by redundancy. For example, in the 1-out-of-4 system, the best-performing instance of DCMAC consistently maintains only component 2, as shown in the belief space plot in Fig.~\ref{fig:belief_occupany_1of4}. This mirrors the behavior of suboptimal JAC instances. In contrast, SARSOP, JAC, and DDQN maintain both components 1 and 2, leading to marginally lower costs as reported in Table~\ref{tab:best_results_kn_infinitehorizon}. DCMAC nevertheless exhibits smaller optimality losses than the other decentralized methods across the \(k\)-out-of-4 systems. 

\begin{figure}
    \centering
    \includegraphics[width=0.76\linewidth]{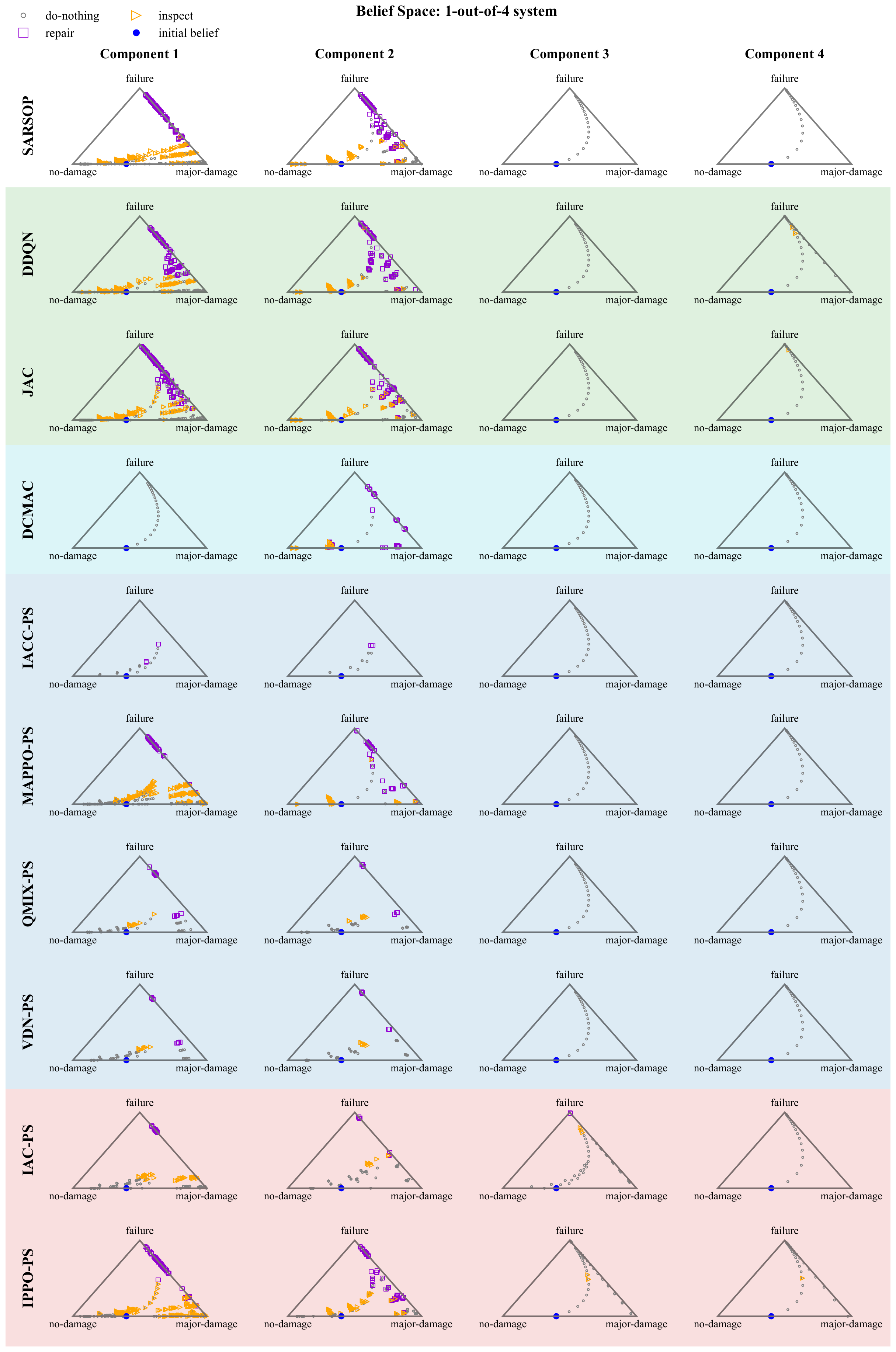}
    \caption{Belief space plots. Using the best-performing training instance of each algorithm, we visualize the learned policies over the belief space. Specifically, we show the actions taken over 100 Monte Carlo rollouts of 20 time steps in the 1-out-of-4 system. Each rollout begins from the initial belief \(b_0=(0.6,0.4,0)\) indicated with a blue dot (\textcolor{blue}{\textbullet}). The SARSOP baseline, JAC, and DDQN coordinate repairs across components 1 and 2 and use inspections to keep beliefs updated; they let components 3 and 4 drift toward failure to exploit redundancy. Although most decentralized algorithms similarly learn to let components 3 and 4 drift toward failure, they differ in how they coordinate interventions on components 1 and 2, resulting in varying optimality gaps as detailed in Table~\ref{tab:best_results_kn_infinitehorizon}.}
    \label{fig:belief_occupany_1of4}
\end{figure}

\subsection{CTDE algorithms}

\subsubsection*{IACC-PS and MAPPO-PS}
\begin{figure}[t]
    \centering
    \includegraphics[width=1\linewidth]{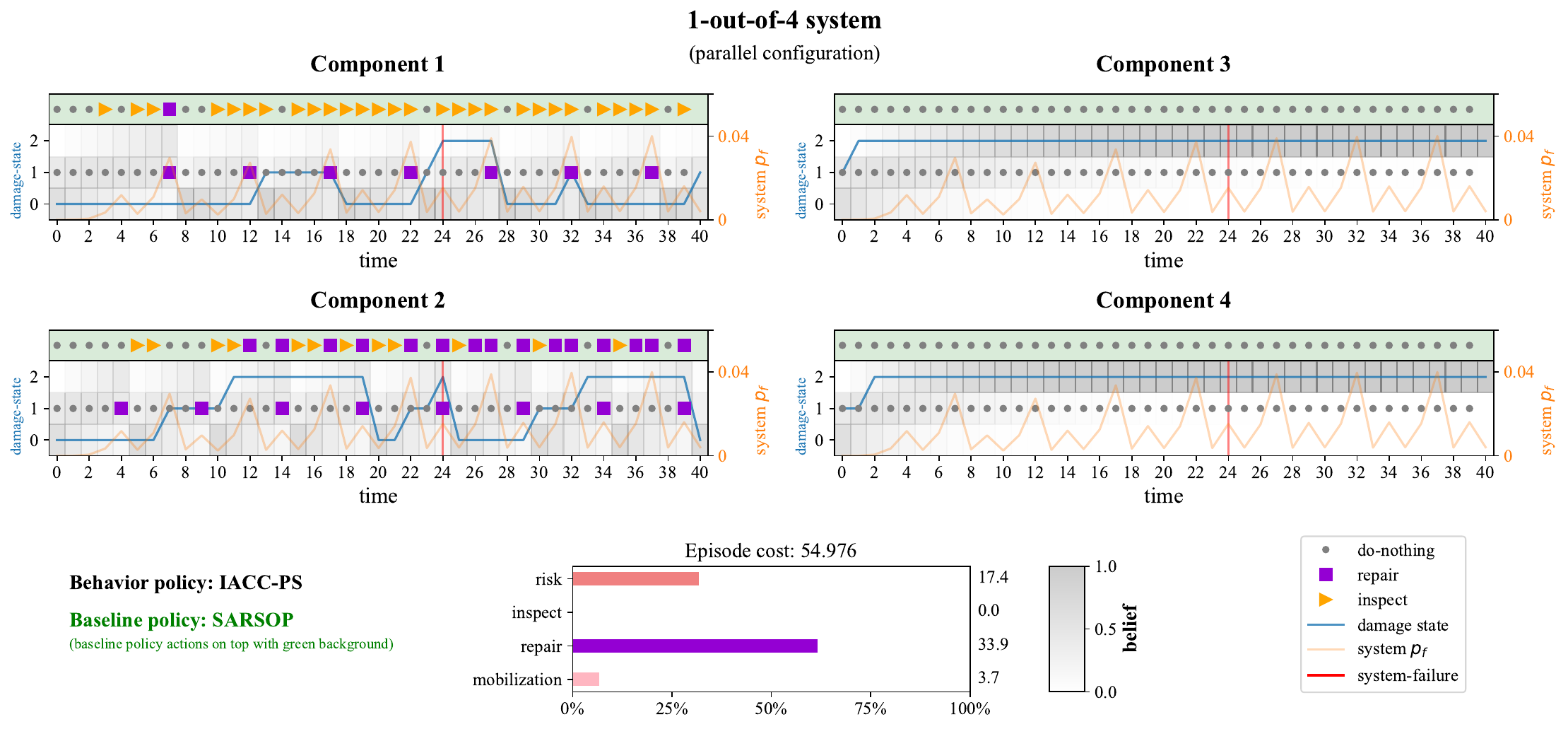}
    \caption{A sample rollout of the best IACC-PS policy in the 1-out-of-4 system for 40 time steps. Interestingly, agent 1 repairs every five time steps starting at $t=7$, while agent 2 follows the same five-step repair cycle starting at $t=4$, resulting in staggered interventions. Neither agent observes the global time step; these actions are selected solely based on their individual components' beliefs.}
    \label{fig:IACC_PS_40}
\end{figure}

Similar to DCMAC, IACC-PS often converges to a suboptimal equilibrium with reward 7 in the Climb Game, achieving a return of \(25 \times 7 = 175\). The centralized critic is unable to efficiently guide coordination, as agents update policies independently with the same advantage estimate. Instances achieving optimal performance are outliers. 
In the 1-out-of-4 system, the best IACC-PS policy is suboptimal but nevertheless exhibits recognizable coordination patterns, as illustrated in Fig.~\ref{fig:IACC_PS_40}. First, only agents 1 and 2 take intervention actions, while agents 3 and 4 remain inactive because their components are more costly to maintain. Second, despite not observing the global time step, agents 1 and 2 follow staggered five-step repair cycles, starting at $t=7$ and $t=4$, respectively.
The repair periodicity can be explained by the belief dynamics induced by the absence of inspections. Although component deterioration is stochastic, each component's belief evolves deterministically between interventions under the stationary transition model. Together with the learned policy, these dynamics induce a Markov chain over beliefs,
\( T_m^{\pi_m}\left(b_m^{\prime} \mid b_m\right)=\sum_{a_m \in A_m} \pi_m\left(a_m \mid b_m\right) T_m\left(b_m^{\prime} \mid b_m, a_m\right) \), whose recurrent trajectories generate the observed repair cycles without explicit access to time. 
Consistent with this interpretation, the beliefs of components 1 and 2 repeatedly traverse the same trajectories before being reset by repair (Fig.~\ref{fig:belief_occupany_1of4}). In contrast, components 3 and 4 consistently drift toward failure. 

MAPPO-PS, the on-policy counterpart of IACC-PS, also fails to reach the optimal return in the Climb Game.
As shown in the belief space plot in Fig.~\ref{fig:belief_occupany_1of4}, it uses inspections and repairs mainly for components 1 and 2, while components 3 and 4 are inactive, similar to IACC-PS. Unlike IACC-PS, which follows a periodic repair strategy, MAPPO-PS exhibits a belief-dependent policy with inspection and repair actions. Notably, MAPPO-PS exhibits greater performance variance in the 1-out-of-4 system than in the other \(k\)-out-of-4 systems. More generally, MAPPO-PS follows the overall trend observed across decentralized algorithms: performance deteriorates as redundancy increases. Consistent patterns are also observed in the additional studies with \(n \in \{2,3\}\) and with an alternate reward model (\ref{appendix:ablation_study}).

\subsubsection*{VDN-PS and QMIX-PS}

\begin{figure}[t]
\centering
\resizebox{\linewidth}{!}{
\begin{tabular}{@{}c@{\hspace{25pt}}c@{\hspace{25pt}}c@{}}

\begin{tabular}{@{}cc|ccc@{}}
\multicolumn{5}{c}{\textbf{Climb Game}} \\
\addlinespace[10pt]
\multicolumn{2}{c|}{} & \multicolumn{3}{c}{Agent 1} \\
\multicolumn{2}{c|}{} & $\mathfrak{a}^1$ & $\mathfrak{a}^2$ & $\mathfrak{a}^3$ \\ \midrule
\multirow{3}{*}{\rotatebox{90}{Agent 2}} &
$\mathfrak{a}^1$ & \(11^*\) & -30 & 0 \\
& $\mathfrak{a}^2$ & -30 & 7 & 6 \\
& $\mathfrak{a}^3$ & 0 & 0 & 5 \\
\end{tabular}

&

\begin{tabular}{@{}cc|ccc@{}}
\multicolumn{5}{c}{\textbf{VDN-PS Decomposition}} \\
\addlinespace[10pt]
 &  & $Q_1(\mathfrak{a}^{1})$ & $Q_1(\mathfrak{a}^{2})$ & $Q_1(\mathfrak{a}^{3})$ \\
 &  & -4.18 & -7.35 & \textcolor{Orange}{6.31} \\ \midrule
$Q_2(\mathfrak{a}^{1})$ & -4.76 & -8.95 & -12.11 & 1.55 \\
$Q_2(\mathfrak{a}^{2})$ & \textcolor{Orange}{-0.33} & -4.51 & -7.67 & \textbf{5.99} \\
$Q_2(\mathfrak{a}^{3})$ & -0.77 & -4.95 & -8.11 & 5.55 \\
\end{tabular}

&

\begin{tabular}{@{}cc|ccc@{}}
\multicolumn{5}{c}{\textbf{QMIX-PS Decomposition}} \\
\addlinespace[10pt]
 &  & $Q_1(\mathfrak{a}^{1})$ & $Q_1(\mathfrak{a}^{2})$ & $Q_1(\mathfrak{a}^{3})$ \\
 &  & \textcolor{Orange}{-0.13} & -5153.65 & -2.88 \\ \midrule
$Q_2(\mathfrak{a}^{1})$ & \textcolor{Orange}{1.19} & \textbf{11.10} & -9.54 & 1.65 \\
$Q_2(\mathfrak{a}^{2})$ & -5412.35 & -9.54 & -9.54 & -9.54 \\
$Q_2(\mathfrak{a}^{3})$ & -1.66 & 1.73 & -9.54 & 1.59 \\
\end{tabular}

\end{tabular}
}
\caption{Value decompositions for the Climb Game. Left: Climb Game rewards. Middle and right: value decompositions learned by the best-performing instances of VDN-PS and QMIX-PS, respectively. Tables show estimated joint action--values, with rows and columns corresponding to the agents’ local actions; orange entries denote the greedy joint action. While VDN-PS assigns the correct maximum to the optimal joint action \((a_3,a_2)\), its linear factorization misrepresents other action utilities. QMIX-PS, despite its monotonic mixing constraint, also misrepresents joint utilities, with optimal performance occurring only in outlier cases.}
\label{fig:value_decomposition}
\end{figure}
The value factorization algorithms VDN-PS and QMIX-PS are unable to correctly capture the joint utility of agents’ actions in the Climb Game. As shown in Fig.~\ref{fig:climb-game}, VDN-PS can reach a maximum return of \(25 \times 6 = 150\) by correctly estimating the return of joint action \((a_3,a_2)\), but Fig.~\ref{fig:value_decomposition} shows that its linear factorization does not predict the return of other action combinations correctly. QMIX-PS, though theoretically more expressive due to its monotonic factorization, also incorrectly factorizes utilities; some instances achieve the optimal return, but only as outliers where the greedy action coincidentally aligns with the optimum~\citep{albrecht_multi-agent_2023}. Thus, both methods demonstrate fundamental limitations in factorizing utilities, even in this simple setting.

As shown in Fig.~\ref{fig:kn_infinite_horizon}, VDN-PS and QMIX-PS achieve optimal performance in the series (4-out-of-4) and series-like (3-out-of-4) systems. We posit that the inductive bias of value factorization, decomposing the joint value into agent-wise contributions, aligns with the structure of series and series-like systems, where most or all components must function simultaneously. However, in the 2-out-of-4 and 1-out-of-4 (parallel) systems, the inductive bias of value factorization is insufficient, as the problem structure is no longer well approximated by additive utilities. This misalignment likely leads to suboptimal performance in redundant systems. We provide intuition for this behavior in~\ref{appendix:value_decomposition_matrix_game}, where we show that for a single-state, non-sequential matrix game, linear and monotonic value decompositions perform optimally in a series configuration due to their inductive bias, but are suboptimal for a parallel configuration.

\subsection{DTDE algorithms}

\subsubsection*{IAC-PS and IPPO-PS}
In the Climb Game, IPPO-PS performs similarly to its CTDE counterpart MAPPO-PS. This aligns with the theoretical findings of~\cite{lyu_centralized_2023}, which show that both centralized and decentralized critics marginalize over other agents, explicitly in the centralized case and implicitly in the decentralized one. The Climb Game is a simplified setting with no state or observation, making such marginalization relatively easy. 
In the 1-out-of-4 system, IPPO-PS exhibits lower variance in performance in comparison to MAPPO-PS. In this case, the best-performing IPPO-PS training instance learns a policy that, similar to MAPPO-PS, concentrates interventions on components 1 and 2 while component 4 remains inactive, as illustrated in Fig.~\ref{fig:belief_occupany_1of4}. IPPO-PS uses a belief-dependent policy with inspection and repair regions, particularly for components 1 and 2, and limited inspection behavior for component 3. Nevertheless, the resulting IPPO-PS policy is slightly less effective than MAPPO-PS in this setting, as reflected by its higher expected costs in Table~\ref{tab:best_results_kn_infinitehorizon}.
IAC-PS, the off-policy variant of IPPO-PS, shows high variability across all \(k\)-out-of-4 systems, a trend also confirmed in the ablation studies (\ref{appendix:ablation_study}). While there are outlier runs that achieve optimal performance in series and series-like configurations, the algorithm rarely learns optimal policies in more difficult settings with redundancies, such as parallel configurations. This is likely due to the off-policy nature of IAC-PS, as also noted in~\cite{papoudakis_benchmarking_2021, lowe_multi-agent_2020}. Since updates rely on past experience, suboptimal trajectories can dominate the replay buffer if optimal actions are not sampled sufficiently.

\begin{table}[t]
\centering
\caption{Best performance of single- and multi-agent DRL algorithms across \(k\)-out-of-4 systems over 10 training instances. Reported values are expected costs, computed as discounted costs via Monte Carlo rollouts, with 95\% confidence intervals in brackets. SARSOP, using one-step look-ahead, serves as the near-optimal baseline. SARSOP expected costs shown in bold are also used to normalize results in Fig.~\ref{fig:kn_infinite_horizon}. Entries marked with (\(\ddagger\)) correspond to deterministic periodic repair policies and exhibit negligible variance.}
\label{tab:best_results_kn_infinitehorizon}
\resizebox{\columnwidth}{!}{%
\begin{tabular}{llcccc}
\toprule
\multicolumn{2}{c}{\textbf{\(k\)-out-of-4 systems}} & \textbf{1-out-of-4 system} & \textbf{2-out-of-4 system} & \textbf{3-out-of-4 system} & \textbf{4-out-of-4 system} \\
\midrule
\multirow{2}{*}{\textbf{SARSOP}} & (LBound, UBound) & (43.67, 51.82) & (113.19, 141.76) & (272.84, 309.06) & (727.36, 755.52) \\
 & Expected cost & \textbf{50.41 (50.31, 50.51)} & \textbf{138.59 (138.25, 138.93)} & \textbf{303.93 (303.18, 304.68)} & \textbf{747.54 (745.92, 749.15)} \\
\midrule
\textbf{Heuristic} & Expected cost & 70.19 (70.12, 70.25) & 159.32 (159.01, 159.63) & 358.51 (357.44, 359.57) & 913.03 (911.36, 914.70) \\
\midrule
\multirow{3}{*}{\textbf{CTCE}} & DDQN & 50.81 (50.71, 50.90) & 141.70 (141.48, 141.92) & 312.61 (311.86, 313.36) & 750.20 (748.62, 751.77) \\
 & JAC & 50.89 (50.81, 50.98) & 143.66 (143.38, 143.94) & 311.95 (311.20, 312.70) & 749.90 (748.34, 751.46) \\
 & DCMAC & 53.83 (53.75, 53.92) & 144.45 (144.22, 144.69) & 311.37 (310.57, 312.17) & 749.57 (747.99, 751.16) \\
\midrule
\multirow{4}{*}{\textbf{CTDE}} & IACC-PS & 53.58$^\ddagger$ & 146.39 (146.10, 146.68) & 310.48 (309.68, 311.28) & 748.78 (747.21, 750.35) \\
 & MAPPO-PS & 54.16 (54.06, 54.27) & 148.06 (147.77, 148.34) & 310.36 (309.58, 311.15) & 748.22 (746.62, 749.81) \\
 & QMIX-PS & 56.71 (56.62, 56.79) & 144.31 (144.03, 144.58) & 309.16 (308.38, 309.94) & 747.53 (745.96, 749.09) \\
 & VDN-PS & 58.18 (58.08, 58.27) & 148.93 (148.67, 149.20) & 309.25 (308.48, 310.01) & 746.64 (745.07, 748.21) \\
\midrule
\multirow{2}{*}{\textbf{DTDE}} & IAC-PS & 57.90 (57.77, 58.02) & 146.67 (146.38, 146.96) & 312.59 (311.77, 313.41) & 756.81 (755.22, 758.40) \\
 & IPPO-PS & 54.75 (54.66, 54.84) & 145.99 (145.73, 146.26) & 310.63 (309.84, 311.43) & 749.16 (747.60, 750.72) \\
\bottomrule
\end{tabular}%
}
\end{table}
\section{Conclusions and future work}
Inspection and maintenance planning in large-scale, partially observable engineering systems necessitates decentralized control to achieve scalability, as centralized approaches become intractable due to the exponential growth of joint state, observation, and action spaces. While multi-agent deep reinforcement learning enables scalable approximation through decentralized policy learning, it also introduces coordination pathologies, in addition to the function-approximation and optimization errors inherent to deep reinforcement learning. These pathologies encumber the learning of optimal policies and give rise to what we term the \emph{price of decentralization}: the relative optimality loss incurred by decentralized control with respect to the centralized optimum.

To quantify this loss, we introduce environments based on \koutofn{} systems with $n\in\{2,3,4\}$, where centralized (near-)optimal baselines remain tractable. 
The environments retain key features of real-world systems, including deterioration, redundancy, inspection and repair actions, partial observability, and component heterogeneity.
These provide controlled benchmarks for assessing solution optimality and also serve as a precursor to larger-scale, computationally intensive benchmarks, where near-optimal references are no longer available.

We evaluate the three main decentralization paradigms: (i) centralized training and centralized execution, (ii) centralized training with decentralized execution, and (iii) fully decentralized training and execution, using nine representative algorithms spanning actor–critic and value-based methods under both on-policy and off-policy training. 
Fully centralized methods, such as JAC and DDQN, which have access to complete system information and predict joint actions, remain susceptible to DRL function-approximation and optimization errors and therefore do not consistently achieve exact-optimal performance, even in small systems. Moreover, they cannot scale to larger systems due to the curse of dimensionality.
Decentralized approaches remain near-optimal in series and series-like systems but develop larger optimality gaps as redundancy increases. Across the decentralized algorithms, the average optimality gap increases systematically with redundancy, from only $0.3\% \pm 0.4\%$ in the 4-out-of-4 series system to $10.3\% \pm 3.9\%$ in the highly redundant 1-out-of-4 parallel system. Although some approaches mitigate this loss more effectively than others, the same trend is observed across a broad set of MADRL algorithms.

Notably, DCMAC, which retains access to global information but decentralizes action selection, also exhibits optimality losses as redundancy increases, rising from approximately 0.3\% to 6.8\% from the 4-out-of-4 to the 1-out-of-4 system, though outperforming fully decentralized approaches. This suggests that access to centralized information alone is insufficient to fully resolve coordination challenges induced by redundancy.
CTDE approaches based on value factorization, such as VDN-PS and QMIX-PS, exhibit a similar trend, with the best value-factorization method reaching an optimality gap of approximately 12.5\% in the highly redundant 1-out-of-4 system. In these methods, the additive inductive bias aligns with the problem structure in series and series-like configurations, but performance degrades in redundant systems where joint utilities cannot be captured by such factorizations; the single-state two-agent reliability game provides analytical intuition for this behavior.

Measuring the price of decentralization requires a (near-)optimal centralized reference, which limited the size of our benchmark systems. Although this paper offers some indications, the precise magnitude of this price in larger systems, where exact solutions are currently intractable, remains an open question for future work.
Another key direction for future work concerns understanding the robustness of learned coordination mechanisms. As highlighted by recent findings on out-of-trajectory generalization~\cite{suau_bad_2023}, agents' policies may exploit spurious correlations in the environment to learn coordination, which generalize poorly and lead to brittle behaviors.
A complementary direction is to develop theoretical insights into the performance of actor--critic algorithms in the presence of redundancy, similar to our analysis of value-factorization methods in this paper.
\section*{CRediT authorship contribution statement}
\textbf{Prateek Bhustali:} Writing -- review \& editing, Writing -- original draft, Visualization, Software, Methodology, Investigation, Formal analysis, Data curation, Conceptualization. \textbf{Pablo G. Morato:} Writing -- review \& editing, Validation, Software, Resources, Methodology, Investigation, Formal analysis. \textbf{Konstantinos G. Papakonstantinou:} Writing -- review \& editing, Validation, Resources, Methodology. \textbf{Charalampos P. Andriotis:} Writing -- review \& editing, Validation, Supervision, Project administration, Resources, Methodology, Funding acquisition, Conceptualization.

\section*{Declaration of generative AI and AI-assisted technologies in the manuscript.}

During the preparation of this manuscript, the author(s) made use of ChatGPT (OpenAI) only for language and text checks, toward improving clarity of presentation and detecting typos. In addition, Codex (OpenAI) was used to support code checks. All outputs generated with the assistance of these tools were reviewed, verified, and edited by the author(s), who take full responsibility for the content of the published article.

\section*{Declaration of competing interest}
The authors declare that they have no known competing financial interests or personal relationships that could have appeared to influence the work reported in this paper.

\section*{Acknowledgments}
Prateek Bhustali and Dr.~Charalampos P. Andriotis gratefully acknowledge support from the TU Delft AI Labs program. This work used the DelftBlue supercomputer, operated by the Delft High Performance Computing Centre~\citep{DHPC2024}.

\appendix
\counterwithin*{figure}{section}
\counterwithin*{table}{section}
\renewcommand{\theHfigure}{\thefigure}
\renewcommand{\theHtable}{\thetable}

\section{Value decomposition in a two-agent reliability matrix game}
\label{appendix:value_decomposition_matrix_game}
This appendix shows how the inductive biases of value decomposition methods affect learned policies in redundant systems. To examine this, we consider a single-step two-agent matrix game derived from a two-component reliability setting. Using the parallel and series reward matrices in Table~\ref{tab:matrix_game_structures}, we show the greedy joint action induced by VDN and QMIX under the individual global maximization principle.

\newcommand{\act}[1]{\ensuremath{\mathtt{#1}}}
\newcommand{\DN}{\act{DN}}
\newcommand{\Rep}{\act{Rep}.}

\newcommand{\MatrixGameTwoByTwo}[4]{
\begingroup
\setlength{\tabcolsep}{8pt}
\renewcommand{\arraystretch}{1.2}
\begin{tabular}{@{}c c c c@{}}
  &  & \multicolumn{2}{c}{Agent 2} \\
  & \multicolumn{1}{c|}{} & \DN & \Rep \\
  \cline{2-4}
  \multirow{2}{*}{\rotatebox{90}{\makebox[3em][c]{Agent 1}}}
    & \multicolumn{1}{c|}{\DN} & #1 & #2 \\
  & \multicolumn{1}{c|}{\Rep}  & #3 & #4 \\
\end{tabular}
\endgroup
}

\begin{proposition}[Limitations of Value Factorization Under Redundancy]
\label{prop:value_factorization_series_parallel_detailed}
Consider a two-agent, single-step matrix game with agents \(m\in\{1,2\}\) and action sets
\(A_m=\{\DN,\Rep\}\), where \(\DN\) denotes do nothing and \(\Rep\) denotes repair. Let
\(r: A_1\times A_2 \to \mathbb{R}\) be the joint reward function, and let \(c_1,c_2>0\) denote the repair costs of agents 1 and 2, respectively. Let \(c_f>0\) denote a failure penalty satisfying \(c_f\gg c_1,c_2\). We assume both components are initially failed and must be repaired to avoid the failure penalty. Assume the reward matrix \(r(a_1,a_2)\) is instantiated as either the parallel or series configuration in Table~\ref{tab:matrix_game_structures}, and that actions are selected according to the individual global maximization (IGM) principle.

\begin{table}[h]
\centering
\caption{Two-agent matrix game used to analyze value decomposition. (a) Generic two-agent reward matrix, where each agent chooses between doing nothing (\DN) and repairing (\Rep). (b) Representation under value decomposition. (c) Parallel-system and (d) series-system reward structures, parameterized by repair costs \(c_1, c_2\) and failure penalty \(c_f\), assuming all components are initially failed.
These cases highlight how redundancy alters optimal coordination and exposes limitations of value-factorization methods.}
\label{tab:matrix_game_structures}
\renewcommand{\arraystretch}{1.3}
\setlength{\tabcolsep}{10pt}

\begin{tabular}{cc}
\begin{minipage}[t]{0.46\linewidth}
\centering
\MatrixGameTwoByTwo
{$r_{\DN,\DN}$}
{$r_{\DN,\Rep}$}
{$r_{\Rep,\DN}$}
{$r_{\Rep,\Rep}$}

\vspace{1em}

(a) Generic joint reward matrix
\end{minipage}
&
\begin{minipage}[t]{0.46\linewidth}
\centering
\begingroup
\setlength{\tabcolsep}{8pt}
\renewcommand{\arraystretch}{1.2}
\begin{tabular}{@{}c c c c@{}}
  &  & \multicolumn{2}{c}{Agent 2} \\
  & \multicolumn{1}{c|}{} & $Q_2(\DN)$ & $Q_2(\Rep)$ \\
  \cline{2-4}
  \multirow{2}{*}{\rotatebox{90}{\makebox[3em][c]{Agent 1}}}%
    & \multicolumn{1}{c|}{$Q_1(\DN)$}
    & $Q(\DN,\DN)$
    & $Q(\DN,\Rep)$ \\
  & \multicolumn{1}{c|}{$Q_1(\Rep)$}
    & $Q(\Rep,\DN)$
    & $Q(\Rep,\Rep)$ \\
\end{tabular}
\endgroup

\vspace{1em}
(b) Value decomposition form
\end{minipage}
\\[5em]
\begin{minipage}[t]{0.46\linewidth}
\centering
\MatrixGameTwoByTwo
{$-c_f$}{$-c_2$}{$-c_1$}{$-(c_1+c_2)$}

\vspace{1em}
(c) Parallel system
\end{minipage}
&
\begin{minipage}[t]{0.46\linewidth}
\centering
\MatrixGameTwoByTwo
{$-c_f$}{$-(c_f+c_2)$}{$-(c_f+c_1)$}{$-(c_1+c_2)$}

\vspace{1em}
(d) Series system
\end{minipage}
\end{tabular}
\end{table}

(a) Under linear value decomposition (VDN),
\[
Q_{\mathrm{tot}}(a_1,a_2)=Q_1(a_1)+Q_2(a_2),
\]
the induced greedy joint action is \((\Rep,\Rep)\) in both series and parallel systems.

(b) Under monotonic value factorization (QMIX),
\[
Q_{\mathrm{tot}}(a_1,a_2)=f(Q_1(a_1),Q_2(a_2)),
\qquad
\frac{\partial f}{\partial Q_m}\ge 0,\ m=1,2,
\]
the induced joint action ordering is identical in both series and parallel systems.
Consequently, both VDN and QMIX align with the optimal coordination structure of series systems, but cannot represent the optimal joint action induced by redundancy in parallel systems.
\end{proposition}

\begin{proof}
We prove parts (a) and (b) separately.

\elsparagraph*{(a) \textbf{Linear value decomposition (VDN)}}

Under VDN, the joint action--value function is approximated as
\[
Q_{\mathrm{tot}}(a_1,a_2)=Q_1(a_1)+Q_2(a_2).
\]
Since this is a single-step game, the Bellman target equals the immediate reward, and learning
minimizes the squared residual over all joint actions,
\begin{equation}
\mathcal{L}_{\mathrm{VDN}}
=
\sum_{(a_1,a_2)\in A_1\times A_2}
\big(Q_1(a_1)+Q_2(a_2)-r(a_1,a_2)\big)^2.
\label{eq:vdn_loss}
\end{equation}

Let
\[
q_1^{\DN}:=Q_1(\DN),\quad
q_1^{\Rep}:=Q_1(\Rep),\quad
q_2^{\DN}:=Q_2(\DN),\quad
q_2^{\Rep}:=Q_2(\Rep),
\]
and abbreviate the rewards as
\(
r_{ij}:=r(a_1=i,a_2=j)
\)
for \(i\in\{\DN,\Rep\}\), \(j\in\{\DN,\Rep\}\).
Then \eqref{eq:vdn_loss} expands to
\begin{align}
\mathcal{L}_{\mathrm{VDN}}
&=
\bigl(q_1^{\DN}+q_2^{\DN}-r_{\DN,\DN}\bigr)^2
+\bigl(q_1^{\DN}+q_2^{\Rep}-r_{\DN,\Rep}\bigr)^2
\nonumber\\
&\quad+
\bigl(q_1^{\Rep}+q_2^{\DN}-r_{\Rep,\DN}\bigr)^2
+\bigl(q_1^{\Rep}+q_2^{\Rep}-r_{\Rep,\Rep}\bigr)^2.
\label{eq:vdn_loss_expanded}
\end{align}

\medskip
\noindent\textbf{Step 1: First-order optimality conditions.}
At any minimizer, \(\nabla \mathcal{L}_{\mathrm{VDN}}=0\).
Differentiating \eqref{eq:vdn_loss_expanded} with respect to each variable gives
\begin{align}
\frac{\partial \mathcal{L}_{\mathrm{VDN}}}{\partial q_1^{\DN}}=0
&\;\Longleftrightarrow\;
2q_1^{\DN}+q_2^{\DN}+q_2^{\Rep}
=
r_{\DN,\DN}+r_{\DN,\Rep},
\label{eq:vdn_eq1}
\\
\frac{\partial \mathcal{L}_{\mathrm{VDN}}}{\partial q_1^{\Rep}}=0
&\;\Longleftrightarrow\;
2q_1^{\Rep}+q_2^{\DN}+q_2^{\Rep}
=
r_{\Rep,\DN}+r_{\Rep,\Rep},
\label{eq:vdn_eq2}
\\
\frac{\partial \mathcal{L}_{\mathrm{VDN}}}{\partial q_2^{\DN}}=0
&\;\Longleftrightarrow\;
q_1^{\DN}+q_1^{\Rep}+2q_2^{\DN}
=
r_{\DN,\DN}+r_{\Rep,\DN},
\label{eq:vdn_eq3}
\\
\frac{\partial \mathcal{L}_{\mathrm{VDN}}}{\partial q_2^{\Rep}}=0
&\;\Longleftrightarrow\;
q_1^{\DN}+q_1^{\Rep}+2q_2^{\Rep}
=
r_{\DN,\Rep}+r_{\Rep,\Rep}.
\label{eq:vdn_eq4}
\end{align}

\medskip
\noindent\textbf{Step 2: Eliminate shared terms to isolate action preferences.}
Subtract \eqref{eq:vdn_eq2} from \eqref{eq:vdn_eq1}:
\begin{align}
\bigl(2q_1^{\DN}+q_2^{\DN}+q_2^{\Rep}\bigr)
-
\bigl(2q_1^{\Rep}+q_2^{\DN}+q_2^{\Rep}\bigr)
&=
\bigl(r_{\DN,\DN}+r_{\DN,\Rep}\bigr)
-
\bigl(r_{\Rep,\DN}+r_{\Rep,\Rep}\bigr),
\nonumber\\
2\bigl(q_1^{\DN}-q_1^{\Rep}\bigr)
&=
r_{\DN,\DN}+r_{\DN,\Rep}
-r_{\Rep,\DN}-r_{\Rep,\Rep},
\nonumber\\
q_1^{\DN}-q_1^{\Rep}
&=
\frac{r_{\DN,\DN}+r_{\DN,\Rep}
-r_{\Rep,\DN}-r_{\Rep,\Rep}}{2}.
\label{eq:vdn_gap1_detailed}
\end{align}
Similarly, subtract \eqref{eq:vdn_eq4} from \eqref{eq:vdn_eq3}:
\begin{align}
\bigl(q_1^{\DN}+q_1^{\Rep}+2q_2^{\DN}\bigr)
-
\bigl(q_1^{\DN}+q_1^{\Rep}+2q_2^{\Rep}\bigr)
&=
\bigl(r_{\DN,\DN}+r_{\Rep,\DN}\bigr)
-
\bigl(r_{\DN,\Rep}+r_{\Rep,\Rep}\bigr),
\nonumber\\
2\bigl(q_2^{\DN}-q_2^{\Rep}\bigr)
&=
r_{\DN,\DN}+r_{\Rep,\DN}
-r_{\DN,\Rep}-r_{\Rep,\Rep},
\nonumber\\
q_2^{\DN}-q_2^{\Rep}
&=
\frac{r_{\DN,\DN}+r_{\Rep,\DN}
-r_{\DN,\Rep}-r_{\Rep,\Rep}}{2}.
\label{eq:vdn_gap2_detailed}
\end{align}

\medskip
\noindent\textbf{Step 3: Parallel and series systems.}
We now evaluate the expressions
\eqref{eq:vdn_gap1_detailed}–\eqref{eq:vdn_gap2_detailed}
for both the parallel and series reward structures.
Observe first that in both systems the failure penalty satisfies
\[
r_{\DN,\DN}=-c_f,
\qquad
r_{\Rep,\Rep}=-(c_1+c_2),
\]
and, crucially, the \emph{difference between asymmetric joint rewards} is identical:
\begin{equation}
r_{\Rep,\DN}-r_{\DN,\Rep}
=
\begin{cases}
(-c_1)-(-c_2) = c_2-c_1, & \text{(parallel)},\\[4pt]
\bigl(-(c_f+c_1)\bigr)-\bigl(-(c_f+c_2)\bigr)=c_2-c_1, & \text{(series)}.
\end{cases}
\label{eq:vdn_asymmetry_invariant}
\end{equation}
As a consequence, the right-hand sides of
\eqref{eq:vdn_gap1_detailed}–\eqref{eq:vdn_gap2_detailed}
coincide for the two systems. Substituting either reward structure yields
\begin{align}
q_1^{\DN}-q_1^{\Rep}
&=
\frac{r_{\DN,\DN}+r_{\DN,\Rep}
-r_{\Rep,\DN}-r_{\Rep,\Rep}}{2}
=
c_1-\frac{c_f}{2},
\label{eq:vdn_gap1_final}
\\[4pt]
q_2^{\DN}-q_2^{\Rep}
&=
\frac{r_{\DN,\DN}+r_{\Rep,\DN}
-r_{\DN,\Rep}-r_{\Rep,\Rep}}{2}
=
c_2-\frac{c_f}{2}.
\label{eq:vdn_gap2_final}
\end{align}
Equivalently,
\[
q_1^{\Rep}-q_1^{\DN}=\frac{c_f}{2}-c_1>0,
\qquad
q_2^{\Rep}-q_2^{\DN}=\frac{c_f}{2}-c_2>0,
\]
where the inequalities follow from the assumption \(c_f\gg c_1,c_2\).
Thus, under the IGM principle, both agents strictly prefer action
\(\Rep\), and the greedy joint action is \((\Rep,\Rep)\)
in \emph{both} the parallel and series systems. This joint action is optimal in the series system but suboptimal in the
parallel system, where repairing only one component avoids failure.

\elsparagraph*{(b) \textbf{Monotonic value factorization (QMIX)}}

QMIX assumes that the joint action--value function factorizes as
\begin{equation}
Q_{\mathrm{tot}}(a_1,a_2)
=
f(Q_1(a_1),Q_2(a_2)),
\qquad
\frac{\partial f}{\partial Q_m}\ge 0,\ m=1,2,
\label{eq:qmix_form}
\end{equation}
so that $Q_{\mathrm{tot}}$ is monotone in each agent’s individual utility.

\medskip
\noindent\textbf{Step 1: Consequence of monotonicity.}
Fix any action $a_2\in A_2$ of agent~2.  
If agent~1 prefers action $a_1$ over $a_1'$, i.e.,
\begin{equation}
Q_1(a_1)\ge Q_1(a_1'),
\label{eq:q1_order}
\end{equation}
then monotonicity of $f$ implies
\begin{equation}
Q_{\mathrm{tot}}(a_1,a_2)
=
f(Q_1(a_1),Q_2(a_2))
\ge
f(Q_1(a_1'),Q_2(a_2))
=
Q_{\mathrm{tot}}(a_1',a_2).
\label{eq:qmix_row_monotonicity}
\end{equation}
Since $a_2$ was arbitrary, the ordering induced by $Q_1$ over
$\{\DN,\Rep\}$ must hold \emph{for every action of agent~2}.
An analogous argument applies symmetrically to agent~2.

\emph{Interpretation:} under QMIX, each agent’s local action ranking must be
globally consistent across all joint actions.

\medskip
\noindent\textbf{Step 2: Series system.}
In a series system, system failure is avoided only if \emph{both} components
are repaired. Hence the optimal joint action is $(\Rep,\Rep)$,
which requires
\begin{equation}
Q_1(\Rep)>Q_1(\DN),
\qquad
Q_2(\Rep)>Q_2(\DN).
\label{eq:qmix_series_pref}
\end{equation}
Substituting \eqref{eq:qmix_series_pref} into
\eqref{eq:qmix_row_monotonicity} implies
\[
Q_{\mathrm{tot}}(\Rep,\Rep)
\ge
Q_{\mathrm{tot}}(\DN,\Rep),
\qquad
Q_{\mathrm{tot}}(\Rep,\Rep)
\ge
Q_{\mathrm{tot}}(\Rep,\DN),
\]
and therefore
\[
Q_{\mathrm{tot}}(\Rep,\Rep)
\ge
Q_{\mathrm{tot}}(a_1,a_2)
\quad\text{for all }(a_1,a_2)\neq(\Rep,\Rep).
\]
Thus, QMIX can represent the optimal policy in the series system.

\medskip
\noindent\textbf{Step 3: Parallel system.}
In a parallel system, repairing both components is unnecessary and suboptimal;
the optimal joint action repairs \emph{only the cheaper component}, i.e.,
either $(\Rep,\DN)$ or $(\DN,\Rep)$.
However, monotonicity enforces a fixed ordering:
if $Q_1(\Rep)>Q_1(\DN)$, then $(\Rep,\cdot)$ must be ranked
above $(\DN,\cdot)$ for \emph{all} actions of agent~2.
Consequently, $(\Rep,\Rep)$ must be ranked above
$(\DN,\Rep)$, contradicting the optimal ordering in the parallel
system.
Therefore, monotonic value factorization cannot exactly represent the joint reward structure of the
parallel system.
\end{proof}

\section{Additional results}

\subsection{Additional figures and tables}
\label{appendix:additional_results}

See Figs.~\ref{fig:suboptimal_JAC}--\ref{fig:sarsop_convergence} and Table~\ref{tab:all_results}.

\begin{figure}[ht]
    \centering
    \includegraphics[width=0.9\linewidth]{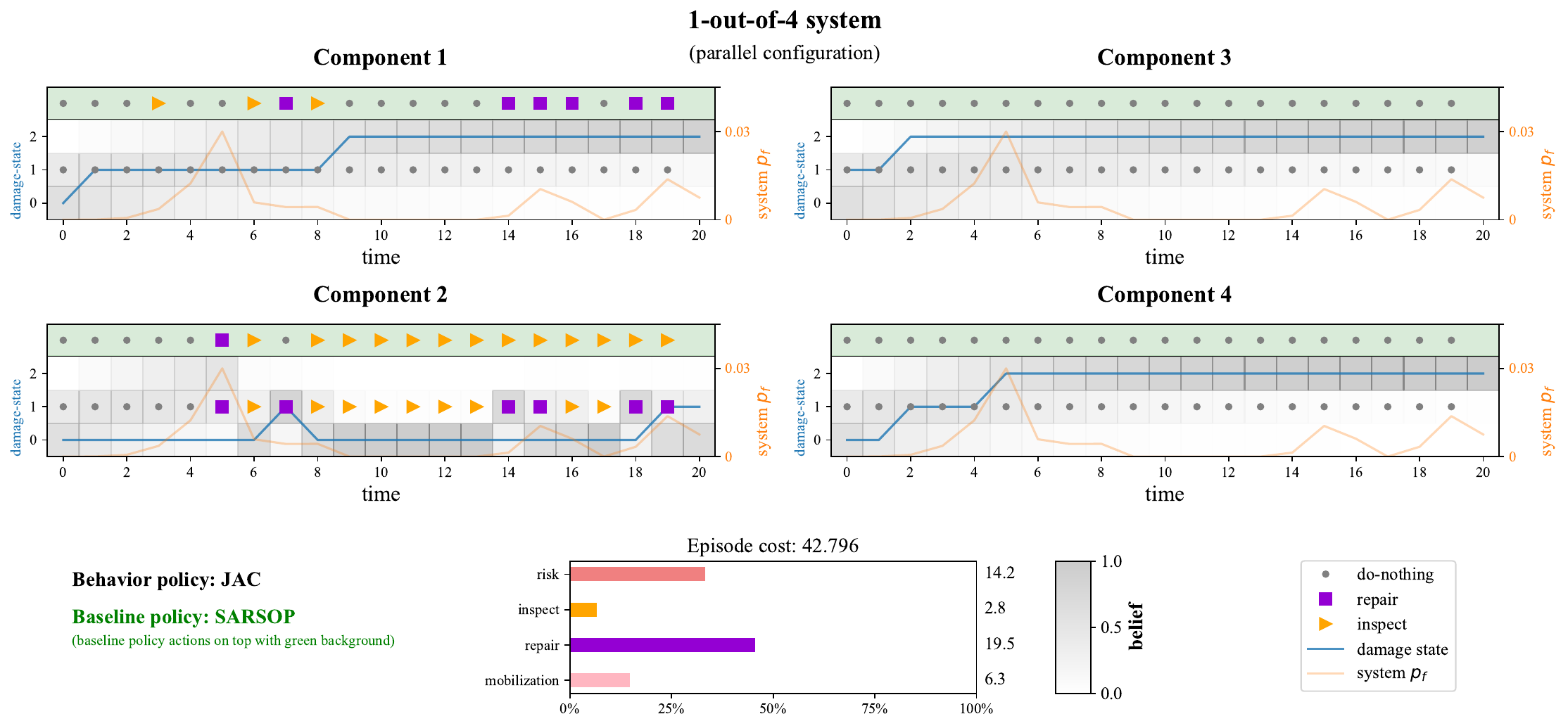}
    \caption{Suboptimal JAC policy in the 1-out-of-4 system. While the optimal JAC and SARSOP policies intervene on components 1 and 2 to maintain system functionality, this learned policy only involves agent 2 taking intervention actions, and all other agents always choose do-nothing. }
    \label{fig:suboptimal_JAC}
\end{figure}

\begin{figure}[ht]
    \centering
    \includegraphics[width=\linewidth]{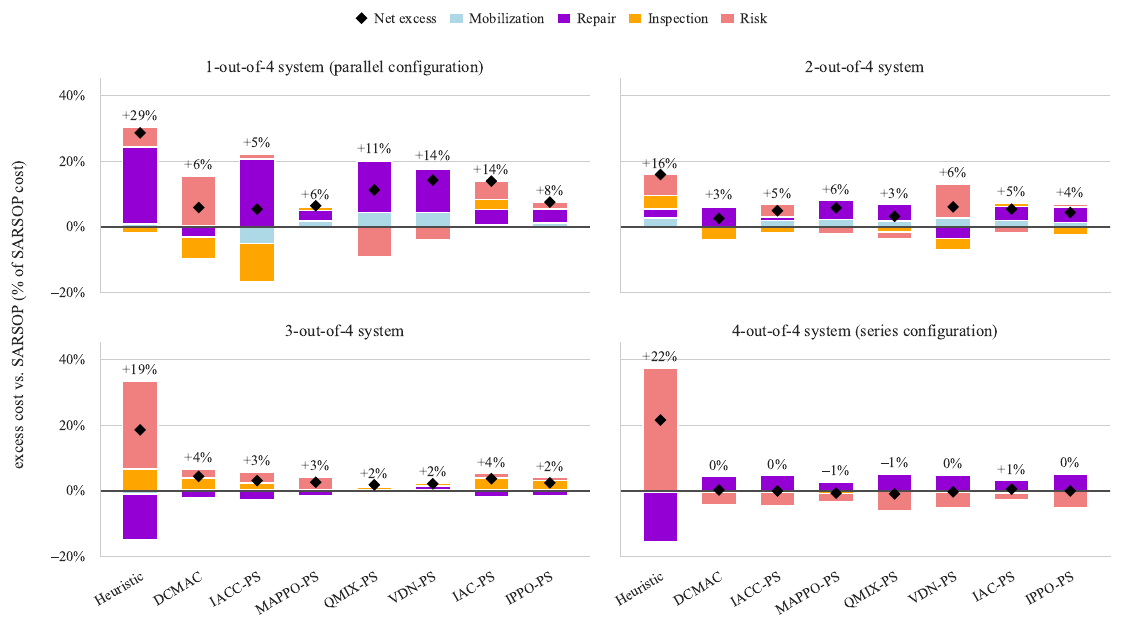}
    \caption{Decomposition of excess cost relative to the SARSOP baseline. For each $k$-out-of-4 system, stacked bars decompose the net excess cost of the best policy (across 10 seeds) into its four constituent terms, i.e., mobilization, repair, inspection, and risk costs, and are expressed as a percentage of the total SARSOP cost. Positive segments indicate cost terms on which a policy spends more than SARSOP; negative segments indicate savings. Black diamonds mark the net excess cost (the sum of all terms), annotated above each bar. The decomposition reveals that policies incur suboptimality through different mechanisms, with no single cost term consistently driving the excess cost across settings or algorithms.}
    \label{fig:excess_cost_vs_sarsop}
\end{figure}

\begin{figure}[tp]
    \centering
    \includegraphics[width=0.85\linewidth]{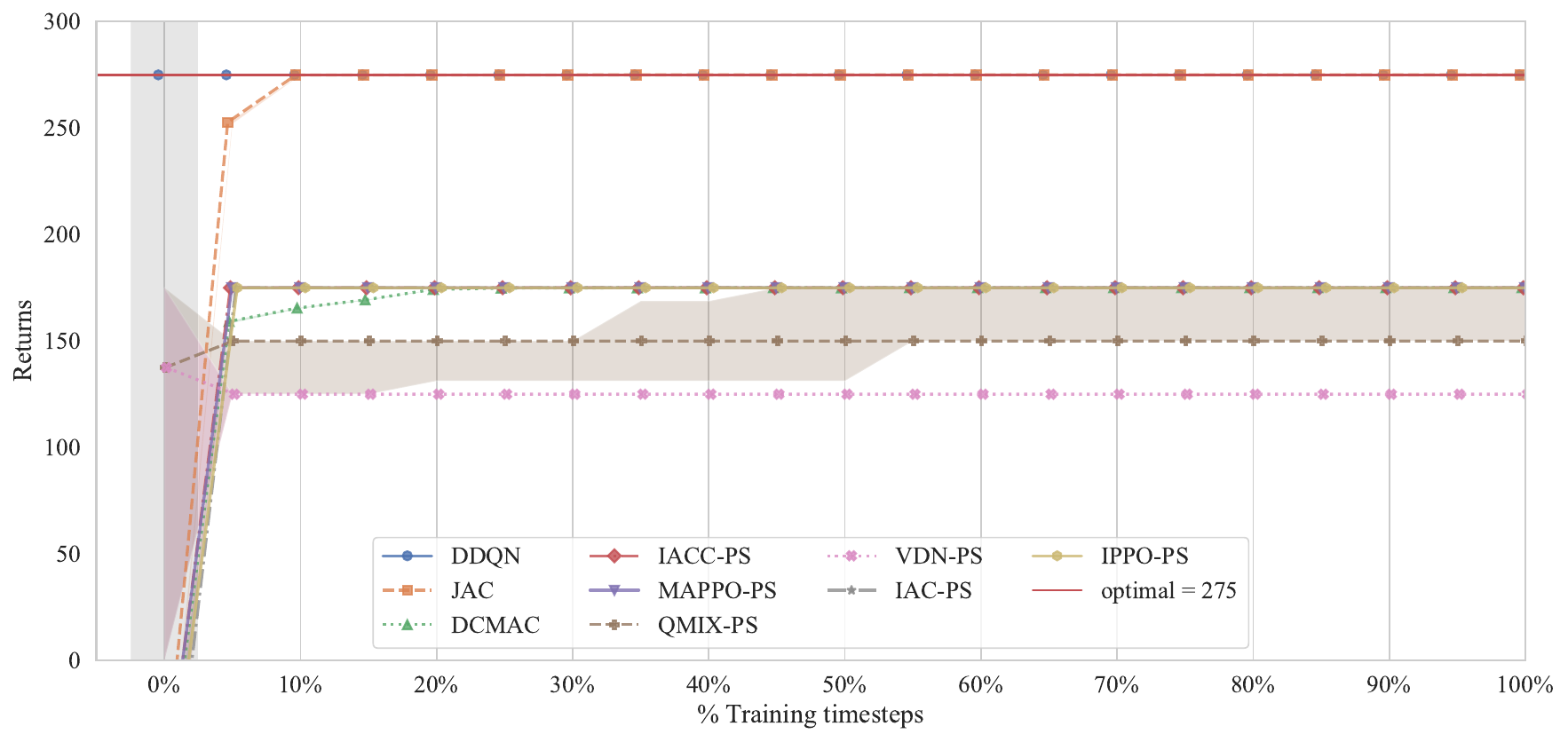}
    \caption{Learning curves of CTCE, CTDE, and DTDE algorithms on the Climb Game. Performance is evaluated periodically during training and shown as median return across seeds, with shaded regions indicating the interquartile range. In this stateless environment, the initial episode is excluded from best-checkpoint selection, since random initializations can yield the optimal action before any learning occurs.}
    \label{fig:learning_curves_climb_game}
\end{figure}

\begin{figure}[tp]
    \centering
    \includegraphics[width=0.95\linewidth]{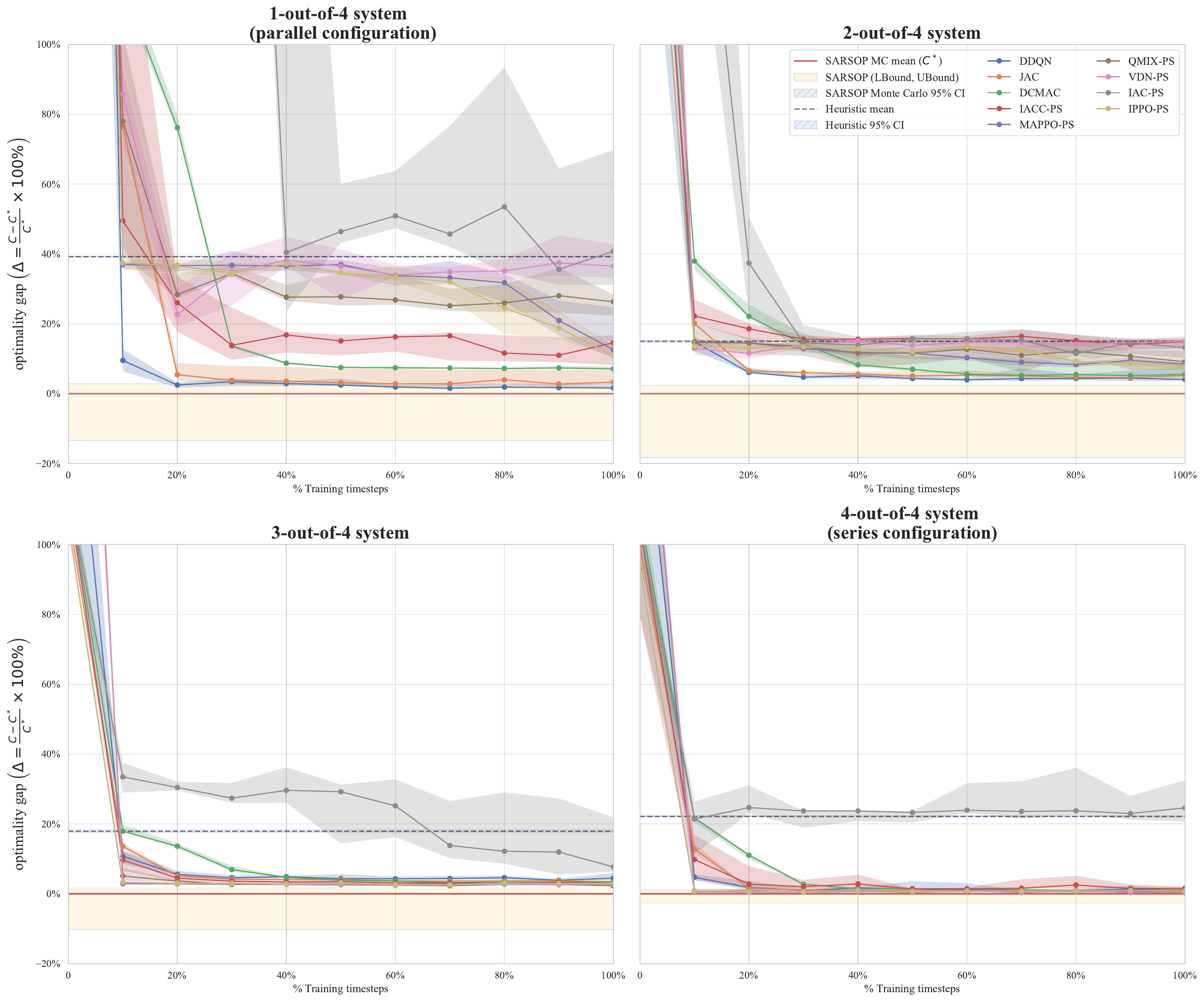}
    \caption{Learning curves of nine algorithms in all infinite-horizon \(k\)-out-of-4 systems summarized using 10 training instances. Performance is evaluated at 11 uniformly spaced checkpoints during training. Costs are normalized with respect to the SARSOP Monte Carlo mean, and the curves show the median performance across seeds, with shaded regions indicating the interquartile range (IQR).}
    \label{fig:learning_curves_kn_infinite}
\end{figure}

\begin{table}[t]
\centering
\caption{The expected costs corresponding to the best checkpoint for each of the 10 training instances, reported in ascending order, across all algorithms and \(k\)-out-of-4 systems. These values correspond to the data summarized in the box plots shown in Fig.~\ref{fig:kn_infinite_horizon}.}
\label{tab:all_results}
\scriptsize
\setlength{\tabcolsep}{3pt}
\renewcommand{\arraystretch}{1.05}
\begin{tabular}{llcccccccccc}
\toprule
  &  & \multicolumn{10}{c}{Training instances} \\
\cmidrule(lr){3-12}
Paradigm & Algorithm & 1 & 2 & 3 & 4 & 5 & 6 & 7 & 8 & 9 & 10 \\
\midrule
\multicolumn{12}{c}{\textbf{1-out-of-4 system (parallel configuration)}} \\
\midrule
\multirow[c]{2}{*}{\textbf{Baseline}} & \textbf{SARSOP (MC)} & \multicolumn{10}{c}{\textbf{50.41} (50.31, 50.51)} \\
  & \textbf{Heuristic} & \multicolumn{10}{c}{\textbf{70.19} (70.12, 70.25)} \\
\midrule
\multirow[c]{3}{*}{\textbf{CTCE}} & DDQN & 50.81 & 50.91 & 50.99 & 51.04 & 51.04 & 51.06 & 51.10 & 51.15 & 51.23 & 51.25 \\
  & JAC & 50.89 & 51.20 & 51.24 & 51.30 & 51.36 & 51.78 & 53.78 & 53.79 & 53.86 & 53.87 \\
  & DCMAC & 53.83 & 53.86 & 53.87 & 53.88 & 53.89 & 53.96 & 53.97 & 53.99 & 54.00 & 54.16 \\
\cmidrule(lr){1-12}
\multirow[c]{4}{*}{\textbf{CTDE}} & IACC-PS & 53.58 & 54.17 & 54.17 & 55.01 & 55.36 & 55.64 & 56.38 & 56.53 & 58.44 & 58.61 \\
  & MAPPO-PS & 54.16 & 54.71 & 54.74 & 55.22 & 56.42 & 57.10 & 59.22 & 61.97 & 64.56 & 65.67 \\
  & QMIX-PS & 56.71 & 57.63 & 58.58 & 59.50 & 60.02 & 60.58 & 62.12 & 62.84 & 63.93 & 64.08 \\
  & VDN-PS & 58.18 & 58.85 & 59.59 & 60.17 & 60.62 & 60.75 & 61.64 & 61.99 & 62.94 & 65.85 \\
\cmidrule(lr){1-12}
\multirow[c]{2}{*}{\textbf{DTDE}} & IAC-PS & 57.90 & 59.47 & 60.79 & 61.08 & 65.07 & 66.78 & 66.88 & 68.45 & 222.23 & 222.24 \\
  & IPPO-PS & 54.75 & 55.43 & 55.60 & 55.62 & 55.90 & 55.92 & 56.43 & 57.90 & 60.23 & 62.58 \\
\midrule
\multicolumn{12}{c}{\textbf{2-out-of-4 system}} \\
\midrule
\multirow[c]{2}{*}{\textbf{Baseline}} & \textbf{SARSOP (MC)} & \multicolumn{10}{c}{\textbf{138.59} (138.25, 138.93)} \\
  & \textbf{Heuristic} & \multicolumn{10}{c}{\textbf{159.32} (159.01, 159.63)} \\
\midrule
\multirow[c]{3}{*}{\textbf{CTCE}} & DDQN & 141.70 & 142.57 & 142.62 & 143.01 & 143.24 & 143.41 & 143.73 & 143.81 & 143.88 & 144.25 \\
  & JAC & 143.66 & 143.70 & 143.77 & 143.88 & 144.04 & 144.94 & 145.20 & 145.37 & 145.88 & 146.53 \\
  & DCMAC & 144.45 & 144.85 & 144.86 & 145.10 & 145.17 & 145.66 & 145.73 & 151.94 & 153.02 & 154.04 \\
\cmidrule(lr){1-12}
\multirow[c]{4}{*}{\textbf{CTDE}} & IACC-PS & 146.39 & 147.07 & 148.50 & 150.34 & 152.08 & 153.79 & 159.51 & 160.00 & 160.39 & 160.43 \\
  & MAPPO-PS & 148.06 & 148.07 & 148.47 & 148.67 & 148.89 & 149.28 & 149.42 & 149.45 & 149.65 & 154.42 \\
  & QMIX-PS & 144.31 & 146.57 & 146.65 & 146.75 & 146.83 & 147.60 & 149.98 & 152.71 & 155.20 & 155.46 \\
  & VDN-PS & 148.93 & 149.51 & 149.76 & 150.45 & 151.25 & 152.31 & 152.53 & 152.84 & 153.89 & 155.38 \\
\cmidrule(lr){1-12}
\multirow[c]{2}{*}{\textbf{DTDE}} & IAC-PS & 146.67 & 147.08 & 149.33 & 150.16 & 151.70 & 154.93 & 157.46 & 157.95 & 158.43 & 172.30 \\
  & IPPO-PS & 145.99 & 147.31 & 147.63 & 148.82 & 149.07 & 150.10 & 150.22 & 150.36 & 150.44 & 150.52 \\
\midrule
\multicolumn{12}{c}{\textbf{3-out-of-4 system}} \\
\midrule
\multirow[c]{2}{*}{\textbf{Baseline}} & \textbf{SARSOP (MC)} & \multicolumn{10}{c}{\textbf{303.93} (303.18, 304.68)} \\
  & \textbf{Heuristic} & \multicolumn{10}{c}{\textbf{358.51} (357.44, 359.57)} \\
\midrule
\multirow[c]{3}{*}{\textbf{CTCE}} & DDQN & 312.61 & 313.16 & 313.45 & 313.72 & 314.12 & 314.48 & 314.67 & 314.86 & 315.41 & 315.86 \\
  & JAC & 311.95 & 312.51 & 312.87 & 313.71 & 313.91 & 314.10 & 314.47 & 314.95 & 316.03 & 316.75 \\
  & DCMAC & 311.37 & 311.41 & 311.89 & 311.99 & 312.56 & 312.65 & 312.99 & 313.04 & 313.06 & 313.40 \\
\cmidrule(lr){1-12}
\multirow[c]{4}{*}{\textbf{CTDE}} & IACC-PS & 310.48 & 310.95 & 311.07 & 311.15 & 311.16 & 312.00 & 312.51 & 313.23 & 314.22 & 383.10 \\
  & MAPPO-PS & 310.36 & 310.47 & 310.68 & 310.99 & 311.03 & 311.32 & 311.57 & 311.67 & 311.67 & 311.89 \\
  & QMIX-PS & 309.16 & 309.29 & 309.41 & 309.57 & 309.62 & 309.67 & 309.73 & 310.23 & 310.41 & 310.48 \\
  & VDN-PS & 309.25 & 309.63 & 309.65 & 309.70 & 309.79 & 310.01 & 310.24 & 310.42 & 310.96 & 311.26 \\
\cmidrule(lr){1-12}
\multirow[c]{2}{*}{\textbf{DTDE}} & IAC-PS & 312.59 & 313.02 & 317.01 & 326.35 & 326.62 & 328.05 & 339.10 & 380.90 & 384.10 & 386.95 \\
  & IPPO-PS & 310.63 & 310.67 & 310.78 & 310.97 & 310.98 & 311.01 & 311.14 & 311.27 & 311.37 & 311.73 \\
\midrule
\multicolumn{12}{c}{\textbf{4-out-of-4 system (series configuration)}} \\
\midrule
\multirow[c]{2}{*}{\textbf{Baseline}} & \textbf{SARSOP (MC)} & \multicolumn{10}{c}{\textbf{747.54} (745.92, 749.15)} \\
  & \textbf{Heuristic} & \multicolumn{10}{c}{\textbf{913.03} (911.36, 914.70)} \\
\midrule
\multirow[c]{3}{*}{\textbf{CTCE}} & DDQN & 750.20 & 750.32 & 750.91 & 751.23 & 751.58 & 752.27 & 752.77 & 752.99 & 753.11 & 754.19 \\
  & JAC & 749.90 & 749.96 & 750.01 & 750.04 & 750.76 & 750.86 & 750.98 & 751.27 & 751.31 & 751.64 \\
  & DCMAC & 749.57 & 749.66 & 751.18 & 751.47 & 752.25 & 753.03 & 753.13 & 753.31 & 753.82 & 753.86 \\
\cmidrule(lr){1-12}
\multirow[c]{4}{*}{\textbf{CTDE}} & IACC-PS & 748.78 & 749.90 & 750.39 & 750.44 & 751.48 & 752.82 & 757.60 & 757.88 & 763.93 & 779.51 \\
  & MAPPO-PS & 748.22 & 749.37 & 749.56 & 749.61 & 749.63 & 749.80 & 750.05 & 750.23 & 750.29 & 750.77 \\
  & QMIX-PS & 747.53 & 748.12 & 748.13 & 748.15 & 748.25 & 748.29 & 748.32 & 748.47 & 748.84 & 749.29 \\
  & VDN-PS & 746.64 & 747.33 & 747.83 & 748.17 & 748.29 & 748.52 & 748.55 & 748.79 & 748.80 & 749.44 \\
\cmidrule(lr){1-12}
\multirow[c]{2}{*}{\textbf{DTDE}} & IAC-PS & 756.81 & 792.81 & 817.22 & 854.56 & 870.68 & 883.67 & 908.00 & 908.67 & 924.88 & 938.68 \\
  & IPPO-PS & 749.16 & 750.02 & 750.08 & 750.16 & 750.22 & 750.24 & 751.23 & 751.48 & 751.72 & 752.37 \\
\bottomrule
\end{tabular}
\end{table}

\FloatBarrier
\begin{center}
    \includegraphics[width=0.75\linewidth]{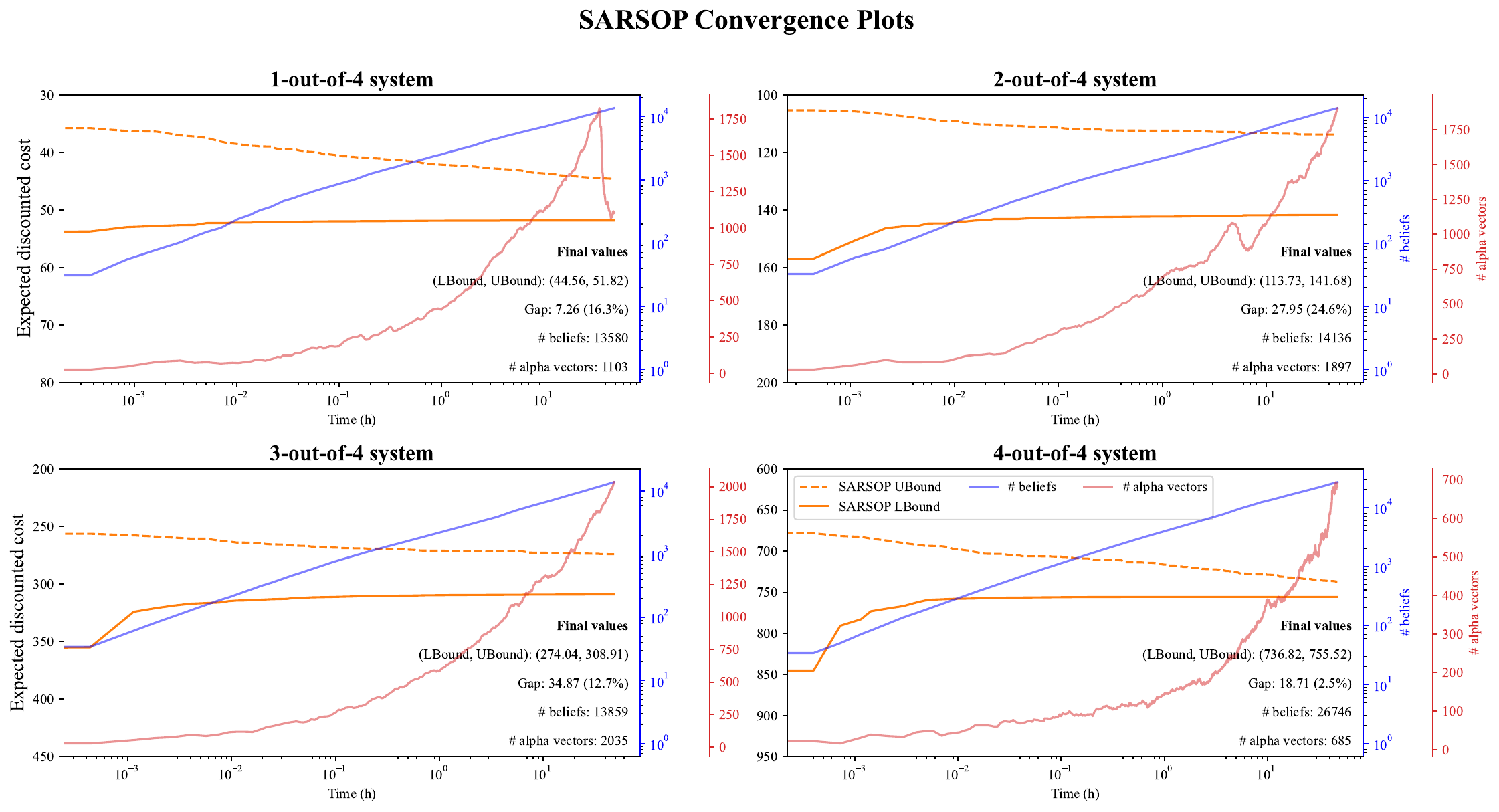}
    \captionof{figure}{SARSOP convergence plots. All settings time out after 48 h, before upper and lower bounds converge. We observe sharp drops in the number of alpha vectors in the 1-out-of-4 and 2-out-of-4 systems, indicating significant pruning.}
    \label{fig:sarsop_convergence}
\end{center}

\subsection{Hyperparameters}

See Table~\ref{tab:hyperparameters}.

\begin{table}[h!]
\centering
\caption{Hyperparameters of single- and multi-agent DRL algorithms used across all \(k\)-out-of-\(n\) systems, and the Climb Game. However, in the Climb Game, we train on-policy and off-policy algorithms for 2.5M time steps.}
\label{tab:hyperparameters}
\resizebox{\textwidth}{!}{
\begin{tabular}{@{}llccccccccc@{}}
\toprule
\multicolumn{2}{l}{Paradigm} &
  \multicolumn{3}{c}{\textbf{CTCE}} &
  \multicolumn{4}{c}{\textbf{CTDE}} &
  \multicolumn{2}{c}{\textbf{DTDE}} \\ \midrule
\multicolumn{2}{l}{Algorithm} &
  \textbf{DDQN} &
  \textbf{JAC} &
  \textbf{DCMAC} &
  \textbf{IACC-PS} &
  \multicolumn{1}{l}{\textbf{MAPPO-PS}} &
  \textbf{VDN-PS} &
  \textbf{QMIX-PS} &
  \textbf{IAC-PS} &
  \multicolumn{1}{l}{\textbf{IPPO-PS}} \\ \midrule
 &
  Episodes (E) &
  \multicolumn{4}{c}{50,000} &
  -- &
  \multicolumn{3}{c}{50,000} &
  -- \\
 &
  Time steps (T) &
  \multicolumn{4}{c}{$E \times t_h = 50,000 \times 50 = 2.5$M} &
  20M &
  \multicolumn{3}{c}{2.5M} &
  20M \\ \midrule
\multirow{6}{*}{\textbf{Architecture}} &
  Actor &
  -- &
  {[}12, 64, 64, 81{]} &
  {[}12, 32, 32, 12{]} &
  {[}7, 32, 32, 3{]} &
  {[}7, 64, 64, 3{]} &
  -- &
  -- &
  {[}7, 32, 32, 3{]} &
  \multicolumn{1}{l}{{[}7, 64, 64, 3{]}} \\
 &
  \# parameters &
  -- &
  10,257 &
  1,868 &
  1,411 &
  4,867 &
  -- &
  -- &
  1,411 &
  4,867 \\
 &
  Critic &
  {[}12, 64, 64, 81{]} &
  {[}12, 64, 64, 1{]} &
  {[}12, 64, 64, 1{]} &
  {[}12, 64, 64, 1{]} &
  {[}12, 64, 64, 1{]} &
  {[}7, 64, 64, 3{]} &
  {[}7, 64, 64, 3{]} &
  {[}7, 64, 64, 1{]} &
  {[}7, 64, 64, 1{]} \\
 &
  \# parameters &
  10,257 &
  5,057 &
  5,057 &
  5,057 &
  5,057 &
  4,867 &
  4,867 &
  4,737 &
  5,057 \\
 &
  \# mixer parameters &
  -- &
  -- &
  -- &
  -- &
  -- &
  -- &
  2945 &
  -- &
  -- \\
 &
  Total \# parameters &
  10,257 &
  15,314 &
  6,925 &
  6,468 &
  9,924 &
  4,867 &
  7,812 &
  6,148 &
  9,924 \\ \midrule
\multirow{9}{*}{\textbf{Learning}} &
  parallel envs &
  1 &
  1 &
  1 &
  1 &
  4 &
  1 &
  1 &
  1 &
  4 \\
 &
  \# env steps &
  1 &
  1 &
  1 &
  1 &
  128 &
  1 &
  1 &
  1 &
  128 \\
 &
  batch size &
  64 &
  64 &
  64 &
  64 &
  128 $\times$ 4 = 512 &
  64 &
  64 &
  64 &
  128 $\times$ 4 = 512 \\
 &
  \# minibatches &
  1 &
  1 &
  1 &
  1 &
  4 &
  1 &
  1 &
  1 &
  4 \\
 &
  optimizer &
  \multicolumn{9}{c}{Adam} \\
 &
  Value function coefficient &
  -- &
  -- &
  -- &
  -- &
  0.5 &
  -- &
  -- &
  -- &
  0.5 \\
 &
  Entropy coefficient &
  -- &
  -- &
  -- &
  -- &
  0.001 &
  -- &
  -- &
  -- &
  0.001 \\
 &
  Max. grad norm &
  -- &
  -- &
  -- &
  -- &
  0.5 &
  -- &
  -- &
  -- &
  0.5 \\
 &
  Value loss clipping &
  -- &
  -- &
  -- &
  -- &
  True &
  -- &
  -- &
  -- &
  True \\ 
   &
  GAE-$\lambda$ &
  -- &
  -- &
  -- &
  -- &
  0.95 &
  -- &
  -- &
  -- &
  0.95 \\ 
     &
  Clip coefficient &
  -- &
  -- &
  -- &
  -- &
  0.1 &
  -- &
  -- &
  -- &
  0.1 \\ 
    &
  Update epochs &
  -- &
  -- &
  -- &
  -- &
  8 &
  -- &
  -- &
  -- &
  8 \\\midrule
\multirow{4}{*}{\textbf{Actor}} &
  initial lr &
  -- &
  1e$-$4 &
  1e$-$4 &
  5e$-$4 &
  2.5e$-$4 &
  -- &
  -- &
  5e$-$4 &
  2.5e$-$4 \\
 &
  end lr &
  -- &
  1e$-$5 &
  1e$-$5 &
  1e$-$5 &
  2.5e$-$4 &
  -- &
  -- &
  1e$-$5 &
  2.5e$-$4 \\
 &
  decay: episodes (E) &
  -- &
  10,000 (E) &
  10,000 (E) &
  20,000 (E) &
  no-decay &
  -- &
  -- &
  10,000 (E) &
  no-decay \\
 &
  decay type &
  \multicolumn{9}{c}{linear} \\ \midrule
\multirow{4}{*}{\textbf{Critic (and Mixer)}} &
  initial lr &
  1e$-$3 &
  5e$-$3 &
  5e$-$3 &
  1e$-$3 &
  2.5e$-$4 &
  1e$-$3 &
  1e$-$3 &
  1e$-$3 &
  2.5e$-$4 \\
 &
  end lr &
  1e$-$4 &
  5e$-$4 &
  5e$-$4 &
  1e$-$4 &
  2.5e$-$4 &
  1e$-$4 &
  1e$-$4 &
  1e$-$4 &
  2.5e$-$4 \\
 &
  decay: episodes (E) &
  10,000 (E) &
  10,000 (E) &
  10,000 (E) &
  20,000 (E) &
  no-decay &
  10,000 (E) &
  10,000 (E) &
  10,000 (E) &
  no-decay \\
 &
  decay type &
  \multicolumn{9}{c}{linear} \\ \midrule
\multirow{4}{*}{\textbf{$\epsilon$-greedy strategy}} &
  $\epsilon$-start &
  \multicolumn{4}{c}{1} &
  -- &
  \multicolumn{3}{c}{1} &
  -- \\
 &
  $\epsilon$-end &
  0.01 &
  0.01 &
  0.001 &
  0.01 &
  -- &
  0.005 &
  0.005 &
  0.01 &
  -- \\
 &
  decay episodes &
  10,000 &
  10,000 &
  20,000 &
  10,000 &
  -- &
  10,000 &
  10,000 &
  20,000 &
  -- \\
 &
  decay type &
  \multicolumn{4}{c}{linear} &
  -- &
  \multicolumn{3}{c}{linear} &
  -- \\ \midrule
\textbf{Replay buffer} &
  time steps &
  20,000 &
  20,000 &
  10,000 &
  10,000 &
  -- &
  20,000 &
  20,000 &
  10,000 &
  -- \\ \midrule
\textbf{Target network reset (Hard)} &
  episodes &
  100 &
  -- &
  -- &
  -- &
  -- &
  100 &
  100 &
  -- &
  -- \\ \bottomrule
\end{tabular}
}
\end{table}

\clearpage
\subsection{Empirical validation using EPyMARL}
\label{appendix:validation}

\begin{figure}[h!]
    \centering
    \includegraphics[width=1\linewidth]{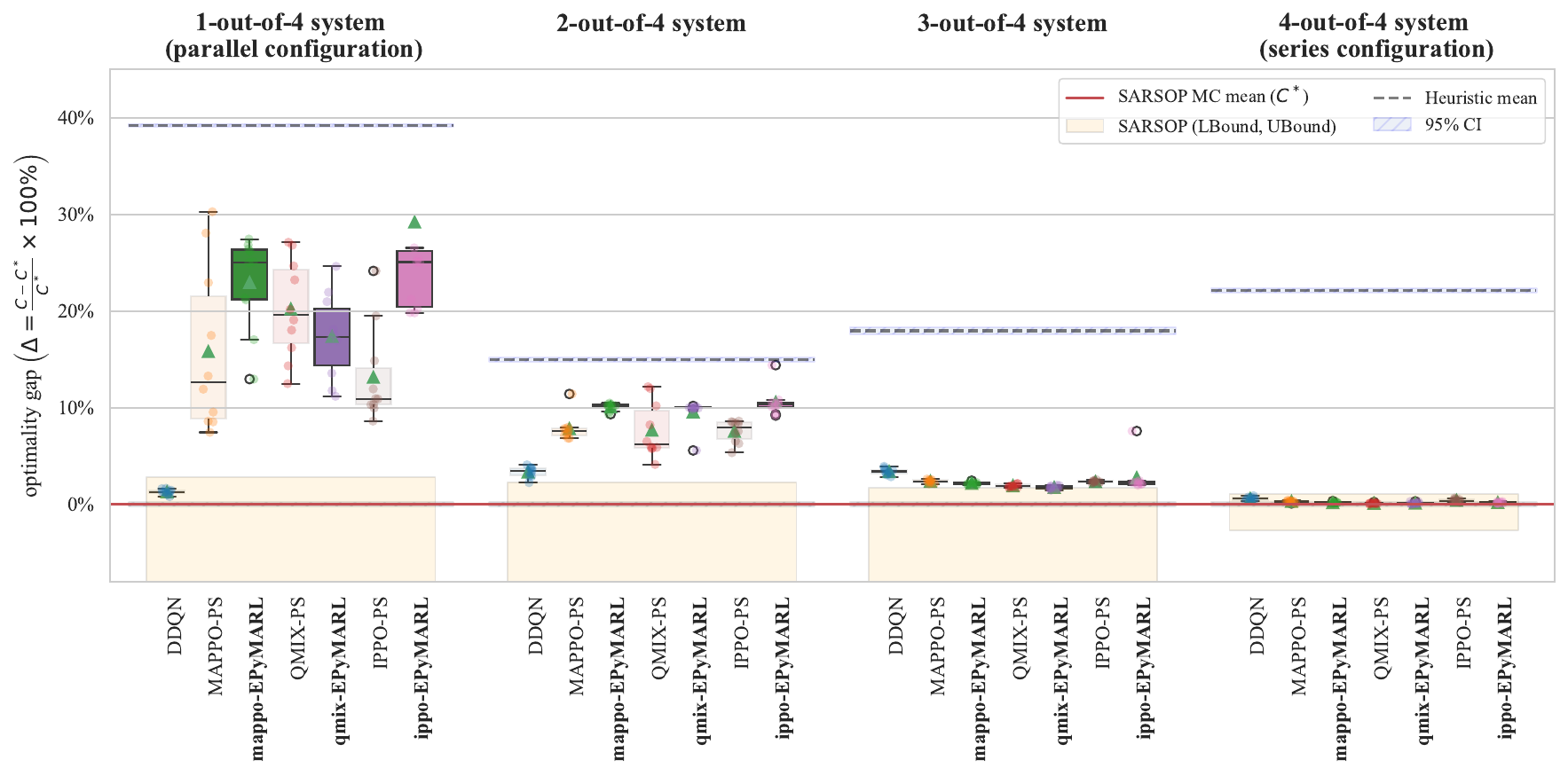}
    \caption{Empirical validation of our MAPPO-PS, QMIX-PS, and IPPO-PS implementations using EPyMARL~\cite{papoudakis_benchmarking_2021} across all \(k\)-out-of-4 systems. Here, we juxtapose the performance of our implementation with the EPyMARL equivalent. Similar to Fig.~\ref{fig:kn_infinite_horizon}, we summarize the performance over 10 training runs, with costs normalized using the SARSOP Monte Carlo mean from Table~\ref{tab:best_results_kn_infinitehorizon}. Note that we also used 100{,}000 Monte Carlo rollouts for the EPyMARL evaluations.}
    \label{fig:epymarl_validation}
\end{figure}

To validate our implementation of MADRL algorithms, we compare the performance of a subset of algorithms (QMIX-PS, MAPPO-PS and IPPO-PS) against EPyMARL \cite{papoudakis_benchmarking_2021}, a standard MADRL library in Python. For this empirical validation, we use the default hyperparameters provided by EPyMARL, without additional tuning. Training and evaluation follow the same protocol outlined in Section~\ref{section:experimental_setup}, including the environment configurations, training horizons, number of seeds, and evaluation procedure. In Fig.~\ref{fig:epymarl_validation}, we observe the same trend as those reported in Fig.~\ref{fig:kn_infinite_horizon} across all \(k\)-out-of-4 systems. In the series setting, MADRL algorithms are able to obtain (near-)optimal performance, but often fail to cooperate effectively in settings with increasing redundancies.

\clearpage
\subsection{Ablation study}
\label{appendix:ablation_study}

This appendix provides additional ablations  that complement the main experiments. They confirm that the conclusions in the main text are robust to changes in system size and reward specification. In Fig.~\ref{fig:ablation_study}, we report results for smaller systems with \(n=2\) and \(n=3\), constructed as consistent subsets of the \(n=4\) base system, as well as an alternate reward model with a reduced failure penalty. For details see Section \ref{section:ablations_and_validation}.

\begin{figure}[h!]
    \centering
    \includegraphics[width=1\linewidth]{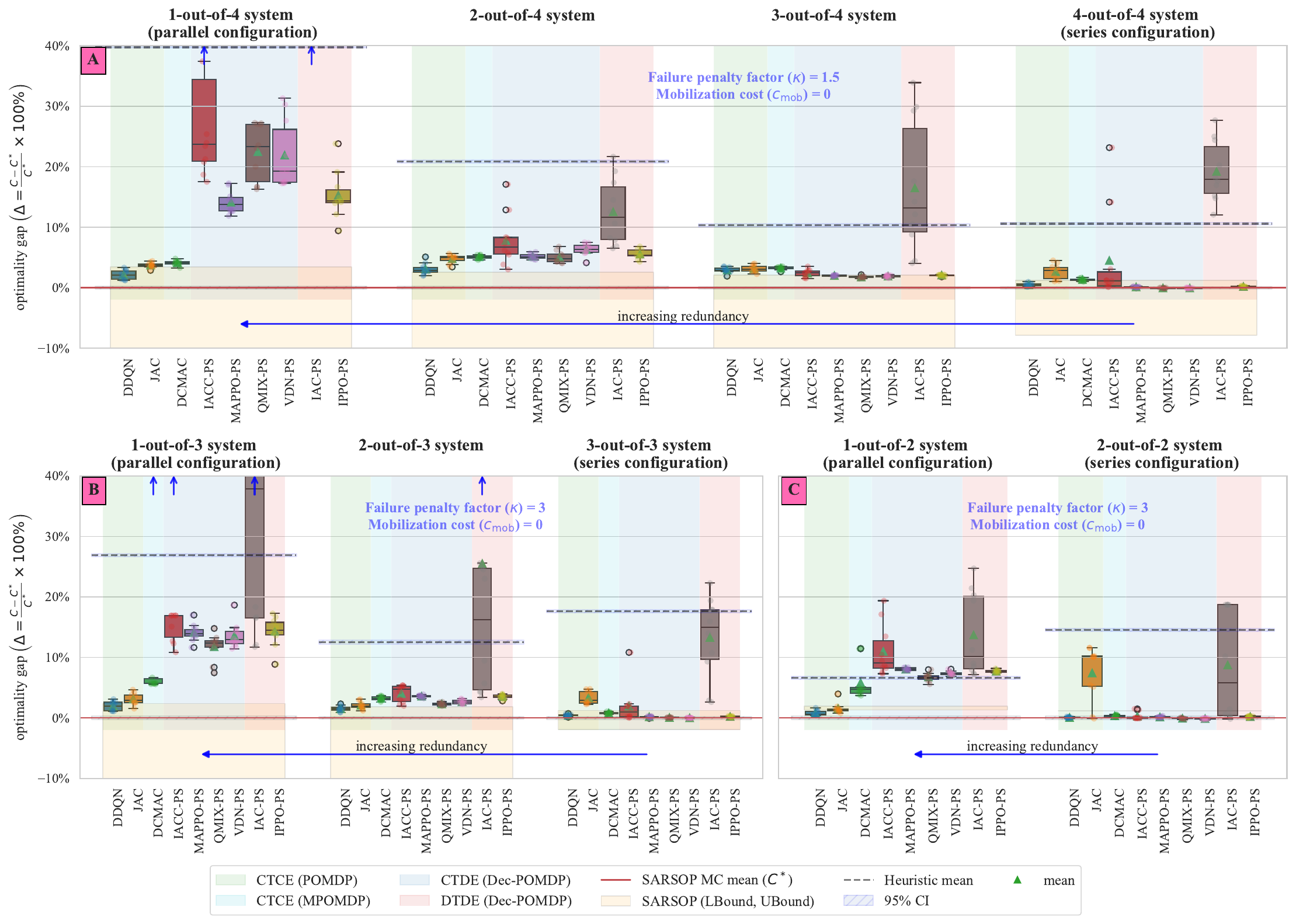}
    \caption{\textbf{Ablation studies}. Performance of MADRL algorithms and heuristic baselines is reported relative to the (near-)optimal SARSOP baseline. The \(n=4\) system with \(\kappa = 3\) and mobilization cost \(=4\) (main text, Fig.~\ref{fig:kn_infinite_horizon}) serves as the base case. In the following ablations, only the specified environment parameters are varied, while algorithm hyperparameters and evaluation procedures remain unchanged. Blue arrows ({\color{blue}$\uparrow$}) indicate entries that exceed the plotted \(y\)-axis limits.
    \textbf{Panel A}: We test robustness by ablating the reward model, reducing the failure penalty factor from \(\kappa = 3\) to \(\kappa = 1.5\) and setting mobilization costs to 0, on the \(n=4\) system across all \(k\).  
    \textbf{Panel B}: We ablate the number of components by considering an \(n=3\) system (with \(\kappa = 3\) and mobilization cost 0) across all \(k\). 
    \textbf{Panel C}: We further ablate the number of components by considering an \(n=2\) system (with \(\kappa = 3\) and mobilization cost 0) across both values of \(k\). Results confirm consistent performance trends, showing that our findings are not limited to a specific reward model.}
    \label{fig:ablation_study}
\end{figure}

\FloatBarrier

\section*{Data availability}
All code, benchmark environments, and trained models needed to reproduce the results are publicly available at \url{https://github.com/prateekbhustali/price-of-decentralization}.

\clearpage 
\bibliographystyle{elsarticle-num}
\setlength{\bibsep}{3pt} 
\interlinepenalty=10000 
\bibliography{cas-refs}

\end{document}